\documentclass{eptcs}
\usepackage{amsmath,amssymb,amsthm,mathtools}
\usepackage{array}
\usepackage[section]{algorithm}
\usepackage{algpseudocode}
\usepackage{booktabs}
\usepackage{enumitem}
\usepackage{graphicx,subcaption}
\usepackage{hyperref}
\usepackage{listings}
\usepackage{multirow}
\usepackage[draft]{tikzit}

\tikzstyle{gate}=[shape=rectangle, text height=1.5ex, text depth=0.25ex, yshift=0.5mm, fill=white, draw=black, minimum height=5mm, yshift=-0.5mm, minimum width=5mm, font={\small}, tikzit category=circuit]
\tikzstyle{2-gate}=[shape=rectangle, text height=1.5ex, text depth=0.25ex, yshift=0.5mm, fill=white, draw=black, minimum height=10mm, yshift=-0.5mm, minimum width=5mm, font={\small}, tikzit category=circuit]
\tikzstyle{3-gate}=[shape=rectangle, text height=1.5ex, text depth=0.25ex, yshift=0.5mm, fill=white, draw=black, minimum height=15mm, yshift=-0.5mm, minimum width=5mm, font={\small}, tikzit category=circuit]
\tikzstyle{4-gate}=[shape=rectangle, text height=1.5ex, text depth=0.25ex, yshift=0.5mm, fill=white, draw=black, minimum height=20mm, yshift=-0.5mm, minimum width=5mm, font={\small}, tikzit category=circuit]
\tikzstyle{5-gate}=[shape=rectangle, text height=1.5ex, text depth=0.25ex, yshift=0.5mm, fill=white, draw=black, minimum height=25mm, yshift=-0.5mm, minimum width=5mm, font={\small}, tikzit category=circuit]
\tikzstyle{Z dot}=[inner sep=0mm, minimum size=2mm, shape=circle, draw=black, fill={rgb,255: red,221; green,255; blue,221}, tikzit category=zx]
\tikzstyle{Z phase dot}=[minimum size=5mm, font={\footnotesize\boldmath}, shape=rectangle, rounded corners=2mm, inner sep=0.2mm, outer sep=-2mm, scale=0.8, tikzit shape=circle, draw=black, fill={rgb,255: red,221; green,255; blue,221}, tikzit draw=blue, tikzit category=zx]
\tikzstyle{X dot}=[Z dot, shape=circle, draw=black, fill={rgb,255: red,255; green,136; blue,136}, tikzit category=zx]
\tikzstyle{X phase dot}=[Z phase dot, tikzit shape=circle, tikzit draw=blue, fill={rgb,255: red,255; green,136; blue,136}, font={\footnotesize\boldmath}, tikzit category=zx]
\tikzstyle{hadamard}=[fill=yellow, draw=black, shape=rectangle, inner sep=0.6mm, minimum height=1.5mm, minimum width=1.5mm, tikzit category=zx]
\tikzstyle{paulibox}=[fill={rgb,255: red,221; green,221; blue,255}, draw=black, shape=rectangle, inner sep=0.6mm, minimum height=5mm, minimum width=5mm, font={\footnotesize}, text height=1.5ex, text depth=0.25ex, tikzit category=zx]
\tikzstyle{vertex}=[inner sep=0mm, minimum size=1mm, shape=circle, draw=black, fill=black, tikzit category=misc]
\tikzstyle{vertex set}=[inner sep=0mm, minimum size=1mm, shape=circle, draw=black, fill=white, font={\footnotesize\boldmath}, tikzit category=misc]
\tikzstyle{small black dot}=[fill=black, draw=black, shape=circle, inner sep=0pt, minimum width=1.2mm, tikzit category=circuit]
\tikzstyle{cnot ctrl}=[fill=black, draw=black, shape=circle, inner sep=0pt, minimum width=1.2mm, tikzit category=circuit]
\tikzstyle{cnot targ}=[fill=white, draw=white, shape=circle, tikzit category=circuit, label={center:$\oplus$}, inner sep=0pt, minimum width=2.1mm, tikzit fill={rgb,255: red,102; green,204; blue,255}, tikzit draw=black]
\tikzstyle{ket}=[fill=white, draw=black, shape=regular polygon, regular polygon sides=3, regular polygon rotate=-30, scale=0.7, inner sep=1pt, tikzit category=circuit, tikzit shape=rectangle, tikzit fill=green]
\tikzstyle{bra}=[fill=white, draw=black, shape=regular polygon, regular polygon sides=3, regular polygon rotate=30, scale=0.7, inner sep=1pt, tikzit category=circuit, tikzit shape=rectangle, tikzit fill=red]
\tikzstyle{scalar}=[shape=rectangle, text height=1.5ex, text depth=0.25ex, yshift=0.5mm, fill=white, draw=black, minimum height=5mm, yshift=-0.5mm, minimum width=5mm, font={\small}]
\tikzstyle{clabel}=[fill=white, draw=none, shape=rectangle, tikzit fill={rgb,255: red,56; green,255; blue,242}, font={\footnotesize}, inner sep=1pt, tikzit category=labels]
\tikzstyle{empty diagram}=[draw={gray!40!white}, dashed, shape=rectangle, minimum width=1cm, minimum height=1cm, tikzit category=misc]
\tikzstyle{amap}=[fill=white, draw=black, shape=NEbox, tikzit category=asymmetric, tikzit fill=yellow, tikzit shape=rectangle]
\tikzstyle{amap conj}=[fill=white, draw=black, shape=NWbox, tikzit category=asymmetric, tikzit fill=green, tikzit shape=rectangle]
\tikzstyle{amap adj}=[fill=white, draw=black, shape=SEbox, tikzit category=asymmetric, tikzit fill=red, tikzit shape=rectangle]
\tikzstyle{amap trans}=[fill=white, draw=black, shape=SWbox, tikzit category=asymmetric, tikzit fill=orange, tikzit shape=rectangle]
\tikzstyle{astate}=[fill=white, draw=black, shape=NEtriangle, tikzit category=asymmetric, tikzit shape=circle, tikzit fill=yellow]
\tikzstyle{astate conj}=[fill=white, draw=black, shape=NWtriangle, tikzit category=asymmetric, tikzit shape=circle, tikzit fill=green]
\tikzstyle{astate adj}=[fill=white, draw=black, shape=SEtriangle, tikzit category=asymmetric, tikzit shape=circle, tikzit fill=red]
\tikzstyle{astate trans}=[fill=white, draw=black, shape=SWtriangle, tikzit category=asymmetric, tikzit shape=circle, tikzit fill=orange]
\tikzstyle{tri}=[fill={rgb,255: red,255; green,136; blue,136}, tikzit fill=red, draw=black, shape=regular polygon, regular polygon sides=3, regular polygon rotate=30, scale=0.7, inner sep=0.2pt, minimum size=7mm, tikzit shape=rectangle]
\tikzstyle{tri left}=[fill={rgb,255: red,255; green,136; blue,136}, tikzit fill=red, draw=black, shape=regular polygon, regular polygon sides=3, regular polygon rotate=210, scale=0.7, inner sep=0.2pt, minimum size=7mm, tikzit shape=rectangle]

\tikzstyle{hadamard edge}=[-, dashed, dash pattern=on 2pt off 0.5pt, thick, draw={rgb,255: red,68; green,136; blue,255}]
\tikzstyle{box edge}=[-, dashed, dash pattern=on 2pt off 0.5pt, thick, draw={rgb,255: red,203; green,192; blue,225}]
\tikzstyle{brace edge}=[-, tikzit draw=blue, decorate, decoration={brace,amplitude=1mm,raise=-1mm}]
\tikzstyle{diredge}=[->]
\tikzstyle{double edge}=[-, double, shorten <=-1mm, shorten >=-1mm, double distance=2pt]
\tikzstyle{gray edge}=[-, {gray!60!white}]
\tikzstyle{pointer edge}=[->, very thick, gray]
\tikzstyle{boldedge}=[-, line width=1.6pt, shorten <=-0.17mm, shorten >=-0.17mm]
\tikzstyle{bidir edge}=[<->, very thick, draw={rgb,255: red,191; green,191; blue,191}]

\input{quantum.tikzdefs}
\newtheorem{definition}{Definition}[section]
\newtheorem{proposition}{Proposition}[section]
\newtheorem{lemma}{Lemma}[section]

\newtheorem{corollary}{Corollary}[section]

\renewcommand{\ket}[1]{\left|#1\right>}
\renewcommand{\bra}[1]{\left<#1\right|}

\algrenewcommand\algorithmicrequire{\textbf{Input:}}
\algrenewcommand\algorithmicensure{\textbf{Output:}}

\title{Efficient Classical Simulation of Low-Rank-Width Quantum Circuits Using ZX-Calculus}
\author{Fedor Kuyanov and Aleks Kissinger}
\def\authorrunning{Kuyanov and Kissinger}
\def\titlerunning{Efficient Classical Simulation of Low-Rank-Width Quantum Circuits Using ZX-Calculus}

\begin{document}

\maketitle

\begin{abstract}
    In this paper, we introduce an algorithm for contracting (i.e. numerically evaluating) ZX-diagrams whose computational complexity is governed by rank-width, a structural graph parameter that behaves nicely under ZX rewrite rules. Our method evaluates a graph-like ZX-diagram in $\tilde O(4^R)$ time, given its rank-decomposition of width $R$. By constructing suitable rank-decompositions, we provide a classical simulation algorithm for quantum circuits whose runtime is bounded by both $\tilde O(2^n)$ and $\tilde O(2^{t/2})$, where $n$ is the number of qubits and $t$ is the number of non-Clifford phase gates. Thus, our method is no slower than either the na\"ive state-vector simulation or the elementary stabiliser decomposition baseline, and in practice can be even faster when the reduced ZX-diagram has low rank-width. We benchmark our simulation algorithm against Quimb, a popular tensor contraction library, and observe substantial reductions in floating-point operations, often by several orders of magnitude, for random and structured non-Clifford circuits and for random ZX-diagrams.
\end{abstract}

\section{Introduction}
\label{intro-section}

Classical simulation of quantum circuits has a variety of applications in designing, verifying, and benchmarking quantum computations. Furthermore, many techniques are transferable to related structures such as tensor networks, which can be applied to solve interesting problems in their own right, e.g. in constraint satisfaction or many-body physics. While it is widely believed that classical simulation of quantum circuits is computationally hard in general, improvements in classical simulation algorithms can significantly affect the feasibility of solving specific problems. For example, a good choice of classical simulation algorithm can make the difference between projected simulation time of 10,000 years~\cite{arute2019quantum} and 5 days~\cite{pan2021simulating}.

While certain families of quantum circuits, such as Clifford circuits, are known to be classically simulable in polynomial time, all known algorithms for simulating computationally universal quantum circuits scale exponentially with respect to some metric of the circuit. However, different techniques rely on different metrics. For example, na\"ive state-vector techniques require space that grows exponentially with the number of qubits, matrix product state simulations grow with the amount of entanglement introduced by the circuit~\cite{vidal2003efficient} (which typically grows with the circuit depth), sum-over-Clifford methods (aka stabiliser decompositions) grow exponentially in the number of non-Clifford gates~\cite{bravyi2016trading,kissinger2022simulating,kissinger2022classical}, and sum-over-path methods grow exponentially in the number of non-diagonal gates~\cite{amy2019towards}.

Tensor network techniques, which generalise many of these methods, represent a quantum circuit (or a more general multilinear map) as a graph of smaller collections of complex numbers called \textit{tensors} which can be composed, or \textit{contracted}, by summing over shared indices. Markov and Shi showed that tensor contraction scales exponentially in the treewidth of the tensor network graph~\cite{markov2008simulating}, so tensor networks that are relatively close to trees can be contracted efficiently. More general tensor networks are highly sensitive to the order in which smaller tensors are contracted together to form larger ones, and finding optimal contraction orders is known to be NP-hard. Various state-of-the-art tensor contraction techniques employ sophisticated heuristics to find graph decompositions that minimise treewidth or related metrics, such as edge/vertex congestion~\cite{gray2021hyper}.

The technique we introduce in this paper is similar in spirit to tensor network contraction methods, in that we start with a special kind of tensor network, namely a \textit{ZX-diagram}, and show that computing the scalar value of a closed ZX-diagram scales with respect to a particular metric of the associated graph, namely its \textit{rank-width}~\cite{oum2003approximation}. Unlike treewidth, rank-width can be small even for highly connected graphs: for example, the fully connected graph on $n$ nodes has treewidth $n - 1$ but rank-width 1.

The key observation is that, unlike for arbitrary tensor networks, the amount of entanglement across any bipartition of nodes in a ZX-diagram scales with the cut-rank of that bipartition. This property was first noted for graph states~\cite{hein2006entanglement}, which are closely related to ZX-diagrams (see e.g.~\cite{duncan2010rewriting}), and can be made explicit by simplifying ZX-diagrams using a set of graphical rewrite rules called the \textit{ZX-calculus}~\cite{coecke2008interacting,kissinger2024pqs}. Most of the ZX-calculus rules simplify ZX-diagrams without increasing the rank-width, thus one can exploit low-rank connectivity to reduce the number of edges between graph components. Using this idea, we provide a tensor contraction routine that evaluates a graph-like ZX-diagram in $\tilde O(4^R)$ time, given a rank-decomposition of width $R$. We furthermore show that, for ZX-diagrams arising from quantum circuits, our routine's complexity is bounded by $\tilde O(2^n)$, where $n$ is the number of qubits, and also by $\tilde O(2^{t/2})$, where $t$ is the number of non-Clifford phase gates, for suitably chosen rank-decompositions. This matches both the scaling of the na\"ive state-vector simulation and the elementary stabiliser decomposition baseline. Our tensor contraction routine is implemented in PyZX.\footnote{\url{https://github.com/zxcalc/pyzx/blob/5f5e409/pyzx/rank\_width.py\#L567}}

As in the case of treewidth, finding decompositions of minimal rank-width is NP-hard. So, we first introduce our methods for constructing rank-decompositions -- see Section~\ref{rw-heur-section}. We then introduce the ZX-diagram contraction algorithm in Section~\ref{sim-section}. Finally, in Section~\ref{bench-section}, we benchmark our methods against the na\"ive tensor contraction routine provided by the ZX-calculus library PyZX~\cite{kissinger2020pyzx}, as well as Quimb~\cite{gray2021hyper}, a state-of-the-art tensor contraction tool. We show substantial reductions in floating-point operations for random CNOT+H+T circuits, two sets of structured non-Clifford circuits, and random ZX-diagrams.

During the preparation of this paper, we became aware of independent work by Codsi and Laakkonen making use of treewidth, rank-width, and stabiliser decompositions to do classical simulation of ZX-diagrams~\cite{codsi2026unifying}. Drawing explicit comparisons between the two approaches and/or combining them is a topic of future work.

\section{Preliminaries}
\label{prelim-section}

\subsection{ZX-calculus}
\label{zxcalc-subsection}

ZX-calculus~\cite{kissinger2024pqs} is a graphical formalism for representing quantum circuits as a specific type of tensor networks, called \emph{ZX-diagrams}. A ZX-diagram consists of building blocks called \emph{spiders}, and there are only two types of them: \emph{Z-} and \emph{X-} spiders.

\begin{definition}[Z- and X-spiders]
    \begin{align*}
        \textrm{Z-spider:} \qquad n \left\{\begin{tikzpicture}
	\begin{pgfonlayer}{nodelayer}
		\node [style=Z phase dot] (0) at (0, 0) {$\alpha$};
		\node [style=none] (1) at (-1, 1) {};
		\node [style=none] (2) at (-1, 0.5) {};
		\node [style=none] (3) at (-1, -1) {};
		\node [style=none] (4) at (1, 1) {};
		\node [style=none] (5) at (1, 0.5) {};
		\node [style=none] (6) at (1, -1) {};
		\node [style=none] (7) at (-0.75, 0) {$\vdots$};
		\node [style=none] (8) at (0.75, 0) {$\vdots$};
	\end{pgfonlayer}
	\begin{pgfonlayer}{edgelayer}
		\draw [bend right=15, looseness=0.75] (0) to (2.center);
		\draw [bend right] (0) to (1.center);
		\draw [bend left] (0) to (3.center);
		\draw [bend left] (0) to (4.center);
		\draw [bend left=15] (0) to (5.center);
		\draw [bend right] (0) to (6.center);
	\end{pgfonlayer}
\end{tikzpicture}\right\} m \, &\wdef \underbrace{\ket{0 \dots 0}}_{m} \underbrace{\bra{0 \dots 0}}_{n} +\,e^{i\alpha} \underbrace{\ket{1 \dots 1}}_{m} \underbrace{\bra{1 \dots 1}}_{n}, \\
        \textrm{X-spider:} \qquad n \left\{\begin{tikzpicture}
	\begin{pgfonlayer}{nodelayer}
		\node [style=X phase dot] (0) at (0, 0) {$\alpha$};
		\node [style=none] (1) at (-1, 1) {};
		\node [style=none] (2) at (-1, 0.5) {};
		\node [style=none] (3) at (-1, -1) {};
		\node [style=none] (4) at (1, 1) {};
		\node [style=none] (5) at (1, 0.5) {};
		\node [style=none] (6) at (1, -1) {};
		\node [style=none] (7) at (0.75, 0) {$\vdots$};
		\node [style=none] (8) at (-0.75, 0) {$\vdots$};
	\end{pgfonlayer}
	\begin{pgfonlayer}{edgelayer}
		\draw [bend left] (1.center) to (0);
		\draw [bend left=15] (2.center) to (0);
		\draw [bend right] (3.center) to (0);
		\draw [bend left] (0) to (4.center);
		\draw [bend left=15] (0) to (5.center);
		\draw [bend right] (0) to (6.center);
	\end{pgfonlayer}
\end{tikzpicture}\right\} m \, &\wdef \underbrace{\ket{+ \dots +}}_{m} \underbrace{\bra{+ \dots +}}_{n} +\,e^{i\alpha} \underbrace{\ket{- \dots -}}_{m} \underbrace{\bra{- \dots -}}_{n}.
    \end{align*}
\end{definition}

The phase $\alpha$ is omitted when $\alpha = 0$. The orthonormal bases $\{ \ket{0}, \ket{1} \}$ and $\{ \ket{{+}}, \ket{{-}} \}$ are the eigenbases of the Pauli $Z$ and $X$ matrices, respectively.

Hadamard gates interchange the Z- and X-bases and often appear in ZX-diagrams. Hence, it is convenient to introduce some special notation for them, either as a small yellow box or a dashed edge called the \textit{Hadamard edge}:
\begin{align}
    \label{eq:had-gate}
    \begin{tikzpicture}
	\begin{pgfonlayer}{nodelayer}
		\node [style=none] (0) at (-1, 0) {};
		\node [style=none] (1) at (1, 0) {};
	\end{pgfonlayer}
	\begin{pgfonlayer}{edgelayer}
		\draw [style=hadamard edge] (0.center) to (1.center);
	\end{pgfonlayer}
\end{tikzpicture}
 \weq \begin{tikzpicture}
	\begin{pgfonlayer}{nodelayer}
		\node [style=hadamard] (0) at (0, 0) {};
		\node [style=none] (1) at (-0.75, 0) {};
		\node [style=none] (2) at (0.75, 0) {};
	\end{pgfonlayer}
	\begin{pgfonlayer}{edgelayer}
		\draw (1.center) to (2.center);
	\end{pgfonlayer}
\end{tikzpicture}
 \wdef \begin{tikzpicture}
	\begin{pgfonlayer}{nodelayer}
		\node [style=gate] (0) at (0, 0) {$H$};
		\node [style=none] (1) at (-1, 0) {};
		\node [style=none] (2) at (1, 0) {};
	\end{pgfonlayer}
	\begin{pgfonlayer}{edgelayer}
		\draw (1.center) to (2.center);
	\end{pgfonlayer}
\end{tikzpicture}
 \weq e^{-\pi i / 4} \, \begin{tikzpicture}
	\begin{pgfonlayer}{nodelayer}
		\node [style=Z phase dot] (0) at (-1, 0) {$\frac{\pi}{2}$};
		\node [style=X phase dot] (1) at (0, 0) {$\frac{\pi}{2}$};
		\node [style=Z phase dot] (2) at (1, 0) {$\frac{\pi}{2}$};
		\node [style=none] (3) at (-1.75, 0) {};
		\node [style=none] (4) at (1.75, 0) {};
	\end{pgfonlayer}
	\begin{pgfonlayer}{edgelayer}
		\draw (3.center) to (4.center);
	\end{pgfonlayer}
\end{tikzpicture}
\end{align}

Every quantum circuit can be compiled into a ZX-diagram, that is, a tensor network consisting of Z/X spiders and regular/Hadamard edges. Note that measurement in the standard basis corresponds to a post-selection with a single-legged X-spider having a parametrised phase $a\pi$, representing the measurement outcome $a \in \{0, 1\}$. Then, by the Born rule, the probability of measuring $a$ equals the absolute value squared of the value of the ZX-diagram.

ZX-calculus offers rewrite rules~\cite[Figure~3.1]{kissinger2024pqs} which can be used for automated circuit simplification and simulation. For instance, they allow us to simulate Clifford circuits in polynomial time, providing an alternative proof for the Gottesman-Knill theorem~\cite{aaronson2004improved}. Notably, every ZX-diagram can be efficiently transformed to a \emph{graph-like form}~\cite[Definition~5.1.7]{kissinger2024pqs}, in which there are only green spiders and Hadamard edges. Using two additional rewrite rules derivable from the ZX-calculus ruleset, called \emph{local complementation} and \emph{pivoting} \cite[Lemmas~5.2.9~and~5.2.11]{kissinger2024pqs}, one can further eliminate some Clifford spiders from the graph-like ZX-diagram and obtain a \emph{reduced ZX-diagram}~\cite[Section~7.6]{kissinger2024pqs}, which in the case of scalar ZX-diagrams consists only of internal non-Clifford spiders and \emph{phase gadgets} -- structures at the top of Figure~\ref{reduced-diagram-figure}. Note that the reduced ZX-diagram can become arbitrarily dense; thus, conventional treewidth-based simulators are inefficient in this case.

\begin{figure}[htbp]
    \centering
    \includegraphics[width=0.7\textwidth]{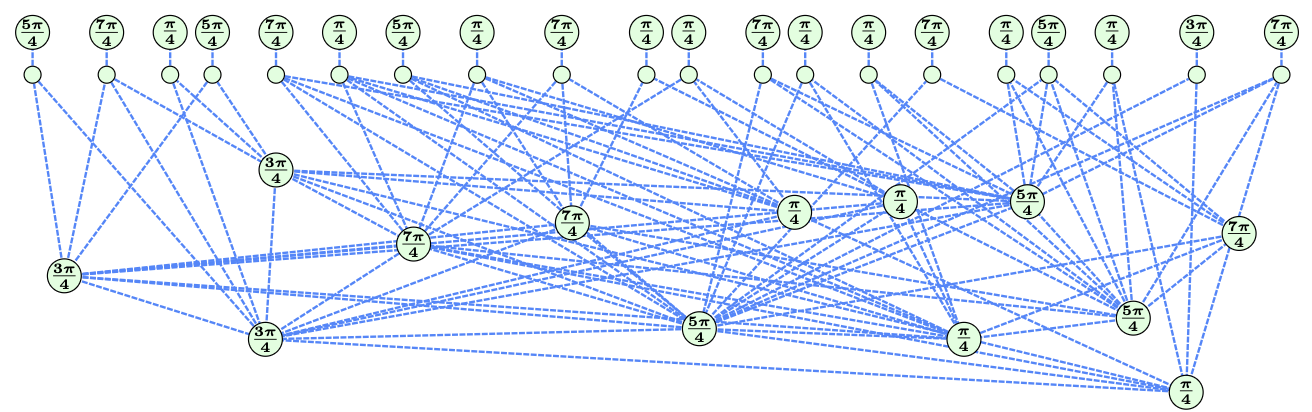}
    \caption{Example of a reduced Clifford+T ZX-diagram.}
    \label{reduced-diagram-figure}
\end{figure}

We also need the notion of parity maps~\cite[Chapter~4]{kissinger2024pqs}. For the two-element finite field $\mathbb F_2$, consider an $\mathbb F_2$-linear map $\phi \colon \mathbb{F}_2^n \to \mathbb{F}_2^m$, mapping the basis states $\ket{x_1 \dots x_n}$ to the basis states $\ket{y_1 \dots y_m}$, where $x_i, y_j \in \{0, 1\}$. We are interested in a graphical representation of $\phi$ in the ZX formalism. Let $A_\phi \in \mathrm{Mat}_{n \times m}(\mathbb{F}_2)$ be the matrix of $\phi$ acting on row vectors of $\mathbb{F}_2^n$. Consider a graph-like ZX-diagram whose spiders form a bipartite graph with biadjacency matrix $A_\phi$, that is, the $i$-th spider on the left is connected to the $j$-th spider on the right iff $(A_\phi)_{ij} = 1$ (see Figure~\ref{parity-map-figure}). One can verify this diagram corresponds to $\phi$ by plugging in the basis states $\ket{x_1 \dots x_n}$, which are single-legged red spiders with phases $x_1 \pi, \dots, x_n \pi$, and applying the standard ZX-calculus rules such as state-copy, spider fusion, and colour change \cite[Chapter~3]{kissinger2024pqs}. Suppose we now have two parity maps $\phi \colon \mathbb{F}_2^n \to \mathbb{F}_2^m$ and $\psi \colon \mathbb{F}_2^m \to \mathbb{F}_2^k$, given by the matrices $A_\phi$ and $A_\psi$, respectively. Recall that the parity matrix of the composition is the product of the parity matrices: $A_{\psi \circ \phi} = A_\phi A_\psi$. This identity can be expressed graphically as in Figure~\ref{parity-map-composition-figure}.

\begin{figure}[htbp]
    \centering
    \begin{minipage}{0.3\textwidth}
        \centering
        \begin{tikzpicture}
	\begin{pgfonlayer}{nodelayer}
		\node [style=Z dot] (0) at (-1, 2) {};
		\node [style=Z dot] (1) at (-1, 1) {};
		\node [style=Z dot] (2) at (-1, -2) {};
		\node [style=Z dot] (3) at (1, 1.5) {};
		\node [style=Z dot] (4) at (1, 0.5) {};
		\node [style=Z dot] (5) at (1, -1.5) {};
		\node [style=none] (6) at (2, 1.5) {};
		\node [style=none] (7) at (2, 0.5) {};
		\node [style=none] (8) at (2, -1.5) {};
		\node [style=none] (9) at (-2, 2) {};
		\node [style=none] (10) at (-2, 1) {};
		\node [style=none] (11) at (-2, -2) {};
		\node [style=none] (12) at (-1.5, -0.25) {$\vdots$};
		\node [style=none] (13) at (1.5, -0.25) {$\vdots$};
		\node [style=none] (14) at (0, 2.75) {$A_\phi$};
		\node [style=none] (15) at (-3, 3) {};
		\node [style=none] (16) at (3, 3) {};
		\node [style=none] (17) at (3, -3) {};
		\node [style=none] (18) at (-3, -3) {};
	\end{pgfonlayer}
	\begin{pgfonlayer}{edgelayer}
		\draw [style=hadamard edge] (3) to (6.center);
		\draw [style=hadamard edge] (4) to (7.center);
		\draw [style=hadamard edge] (5) to (8.center);
		\draw (9.center) to (0);
		\draw (10.center) to (1);
		\draw (11.center) to (2);
		\draw [style=hadamard edge] (0) to (3);
		\draw [style=hadamard edge] (1) to (3);
		\draw [style=hadamard edge] (1) to (4);
		\draw [style=hadamard edge] (2) to (5);
		\draw [style=hadamard edge] (1) to (5);
		\draw [style=hadamard edge] (2) to (4);
	\end{pgfonlayer}
\end{tikzpicture}
        \caption{a parity map}
        \label{parity-map-figure}
    \end{minipage}
    \begin{minipage}{0.6\textwidth}
        \centering
        \begin{tikzpicture}
	\begin{pgfonlayer}{nodelayer}
		\node [style=Z dot] (0) at (-8, 2) {};
		\node [style=Z dot] (1) at (-8, 1) {};
		\node [style=Z dot] (2) at (-8, -2) {};
		\node [style=Z dot] (3) at (-6, 1.5) {};
		\node [style=Z dot] (4) at (-6, 0.5) {};
		\node [style=Z dot] (5) at (-6, -1.5) {};
		\node [style=none] (9) at (-9, 2) {};
		\node [style=none] (10) at (-9, 1) {};
		\node [style=none] (11) at (-9, -2) {};
		\node [style=none] (12) at (-8.5, -0.25) {$\vdots$};
		\node [style=none] (13) at (-5, -0.25) {$\vdots$};
		\node [style=none] (14) at (-7, 2.75) {$A_\phi$};
		\node [style=Z dot] (21) at (-4, 1.5) {};
		\node [style=Z dot] (22) at (-4, 0.5) {};
		\node [style=Z dot] (23) at (-4, -1.5) {};
		\node [style=Z dot] (24) at (-2, 2) {};
		\node [style=Z dot] (25) at (-2, 1) {};
		\node [style=Z dot] (26) at (-2, -2) {};
		\node [style=none] (27) at (-1, 2) {};
		\node [style=none] (28) at (-1, 1) {};
		\node [style=none] (29) at (-1, -2) {};
		\node [style=none] (30) at (-1.5, -0.25) {$\vdots$};
		\node [style=none] (31) at (-3, 2.75) {$A_\psi$};
		\node [style=none] (32) at (0, 0) {$\propto$};
		\node [style=none] (33) at (1, 1) {};
		\node [style=none] (34) at (1, 2) {};
		\node [style=none] (35) at (1, -2) {};
		\node [style=Z dot] (36) at (2, 2) {};
		\node [style=Z dot] (37) at (2, 1) {};
		\node [style=Z dot] (38) at (2, -2) {};
		\node [style=Z dot] (39) at (4, 2) {};
		\node [style=Z dot] (40) at (4, 1) {};
		\node [style=Z dot] (41) at (4, -2) {};
		\node [style=none] (42) at (5, 2) {};
		\node [style=none] (43) at (5, 1) {};
		\node [style=none] (44) at (5, -2) {};
		\node [style=none] (45) at (1.5, -0.25) {$\vdots$};
		\node [style=none] (46) at (4.5, -0.25) {$\vdots$};
		\node [style=none] (47) at (3, 2.75) {$A_\phi A_\psi$};
		\node [style=none] (48) at (-10, 3) {};
		\node [style=none] (49) at (-10, -3) {};
		\node [style=none] (50) at (6, -3) {};
		\node [style=none] (51) at (6, 3) {};
	\end{pgfonlayer}
	\begin{pgfonlayer}{edgelayer}
		\draw (9.center) to (0);
		\draw (10.center) to (1);
		\draw (11.center) to (2);
		\draw [style=hadamard edge] (0) to (3);
		\draw [style=hadamard edge] (1) to (3);
		\draw [style=hadamard edge] (1) to (4);
		\draw [style=hadamard edge] (2) to (5);
		\draw [style=hadamard edge] (1) to (5);
		\draw [style=hadamard edge] (2) to (4);
		\draw [style=hadamard edge] (3) to (21);
		\draw [style=hadamard edge] (4) to (22);
		\draw [style=hadamard edge] (5) to (23);
		\draw [style=hadamard edge] (21) to (24);
		\draw [style=hadamard edge] (22) to (24);
		\draw [style=hadamard edge] (22) to (25);
		\draw [style=hadamard edge] (23) to (25);
		\draw [style=hadamard edge] (23) to (26);
		\draw [style=hadamard edge] (24) to (27.center);
		\draw [style=hadamard edge] (25) to (28.center);
		\draw [style=hadamard edge] (26) to (29.center);
		\draw (34.center) to (36);
		\draw (33.center) to (37);
		\draw (35.center) to (38);
		\draw [style=hadamard edge] (36) to (39);
		\draw [style=hadamard edge] (36) to (40);
		\draw [style=hadamard edge] (37) to (40);
		\draw [style=hadamard edge] (37) to (41);
		\draw [style=hadamard edge] (38) to (41);
		\draw [style=hadamard edge] (38) to (40);
		\draw [style=hadamard edge] (39) to (42.center);
		\draw [style=hadamard edge] (40) to (43.center);
		\draw [style=hadamard edge] (41) to (44.center);
	\end{pgfonlayer}
\end{tikzpicture}
        \caption{composition of parity maps}
        \label{parity-map-composition-figure}
    \end{minipage}
\end{figure}

\subsection{GFlow}
\label{gflow-subsection}

We use the notion of \emph{extended gflow} to bound the rank-width of certain ZX-diagrams. In general, extended gflow provides sufficient conditions under which a measurement pattern on a graph state, with single-qubit measurements in the three Pauli planes, can be executed deterministically. Although originally developed for the one-way model of \emph{measurement-based quantum computing} (MBQC)~\cite{browne2007generalized}, it has since found other applications, most notably in performing efficient circuit extraction from ZX-diagrams that have been simplified using the rules of ZX-calculus~\cite{backens2021there,simmons2021relating}. In particular, extended gflow always exists for reduced ZX-diagrams originating from quantum circuits, and there is a polynomial-time algorithm for finding it \cite[Theorems~3.11~and~4.21]{backens2021there}.

Extended gflow is defined on an \emph{open graph}, that is, a triple $(G, I, O)$ where $G = (V, E)$ is an undirected graph, and $I, O \subset V$ are respectively called input and output vertices. Extended gflow consists of a \emph{correction function} $g$ associating each non-output vertex $v$ with its correction set (in MBQC, it is used to handle the measurement outcomes of $v$), a strict partial order $\prec$ over $V$ denoting the vertex measurement order ($i \prec j$ if $i$ is measured before $j$), and a \emph{labelling function} $\lambda$ assigning to each non-output vertex one of the three Pauli planes (XY, YZ, or XZ).

\begin{definition}[Extended gflow]
    \label{ext-gflow-definition}
    Given an open graph $(G, I, O)$, an \emph{extended gflow} is a triple $(g, \prec, \lambda)$, where $g \colon V(G) \setminus O \to \mathcal{P}(V(G) \setminus I) \setminus \{\varnothing\}$, $\prec$ is a strict partial order over $V(G)$, $\lambda \colon V(G) \setminus O \to \{\mathrm{XY}, \mathrm{XZ}, \mathrm{YZ}\}$, which satisfy the following conditions:
    \begin{enumerate}
        \item if $j \in g(i)$ then $i \preceq j$,
        \item if $j \in \mathrm{Odd}(g(i))$ then $i \preceq j$,
        \item if $\lambda(i) = \mathrm{XY}$ then $i \not\in g(i)$ and $i \in \mathrm{Odd}(g(i))$,
        \item if $\lambda(i) = \mathrm{XZ}$ then $i \in g(i)$ and $i \in \mathrm{Odd}(g(i))$,
        \item if $\lambda(i) = \mathrm{YZ}$ then $i \in g(i)$ and $i \not\in \mathrm{Odd}(g(i))$,
    \end{enumerate}
    where $\mathrm{Odd}(S) := \{u \in V(G) \colon |N(u) \cap S| \equiv 1 \pmod 2\}$.
\end{definition}

\subsection{Rank-width}
\label{rw-subsection}

Rank-width~\cite{oum2003approximation} is a graph measure that captures the behaviour of its cut-ranks over $\mathbb{F}_2$. In particular, rank-width is the lowest possible width of the graph's \emph{rank-decomposition}, which is defined as a recursive partitioning process of its vertices (called \emph{branch-decomposition}) along with cut-ranks assigned to each edge of the decomposition tree. Rank-width is closely related to other structural graph parameters such as clique-width~\cite{courcelle1993handle} and treewidth~\cite{robertson1984graph}. Before defining rank-width, we must introduce the following notions of \emph{branch-decomposition} and \emph{cut-rank}.

\begin{definition}
    Let $M$ be a finite set, called the \emph{ground set} of the decomposition. A \emph{branch-decomposition} is a pair $(T, L)$, where $T$ is a tree with each non-leaf node having degree 3, and $L$ is the bijection from $M$ to the leaves of $T$.
\end{definition}

\begin{definition}
    Given an undirected graph $G = (V, E)$ and a subset $X \subset V$, the \emph{cut-rank} $\rho_G(X)$ is the rank of the adjacency matrix between $X$ and $V \setminus X$, viewed as a matrix over $\mathbb{F}_2$.
\end{definition}

Combining the two notions above, we get the definition of the \emph{rank-decomposition} and \emph{rank-width}. Refer to Figure~\ref{rank-decomp-figure} for illustration.

\begin{definition}
    Let $G = (V, E)$ be an undirected graph. A \emph{rank-decomposition} is a branch-decomposition with the ground set $V$, where each edge $(u, v)$ of the tree is assigned the cut-rank $\rho_{u, v} := \rho_G(X)$, where $X$ is the set of leaves in the same component as $u$ after deleting $(u, v)$. The \emph{width} of the rank-decomposition is the maximum $\rho_{u, v}$ over all edges $(u, v)$.
\end{definition}

\begin{definition}
    The \emph{rank-width} of an undirected graph $G$, written as $\mathrm{rw}(G)$, is the minimum width of its rank-decomposition.
\end{definition}

\begin{figure}[htbp]
    \centering
    \includegraphics[width=0.6\textwidth]{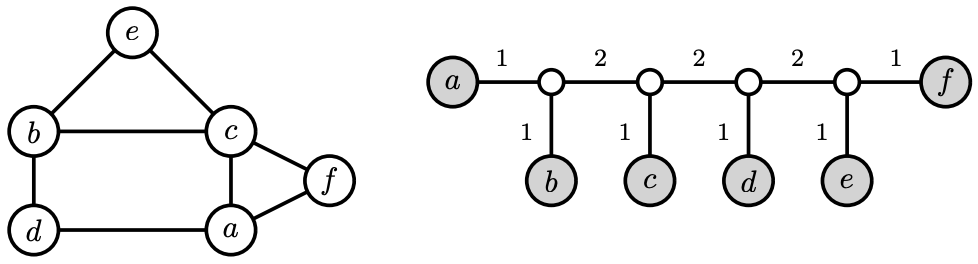}
    \caption{\cite[Figure~3.2]{nouwt2022simulated} Optimal rank-decomposition of $G$, showing $\mathrm{rw}(G) = 2$.}
    \label{rank-decomp-figure}
\end{figure}

In fact, the rank-decomposition presented in Figure~\ref{rank-decomp-figure} has a special property: its cut-ranks are taken over the partitions that exhibit a linear structure. This allows for the following simplified notion of \emph{linear rank-decomposition}.

\begin{definition}
    A rank-decomposition is called \emph{linear} if the non-leaf nodes of its branch-decomposition form a path graph. Consequently, a linear rank-decomposition is given by a permutation $(p_i)_{i=1,\dots,|V|}$ of the graph vertices, with the cut-ranks $r_i$ taken between $\{p_1, \dots, p_i\}$ and $\{p_{i+1}, \dots, p_{|V|}\}$ for each $i \in [1; |V| - 1]$.
\end{definition}

Let us discuss some hardness results related to rank-width. In \cite{oum2017rw}, it is shown that computing the rank-width of a given graph is NP-hard, and the decision problem of determining whether the rank-width does not exceed $k$ is NP-complete. So, one can only hope to compute the rank-width approximately. In \cite{beyss2013fast}, there is an attempt to calculate upper and lower bounds for rank-width. The algorithm for the upper bound iteratively improves the decomposition using ``a mix of greedy and random decisions''. Another approach~\cite{kuyanov2025rw,nouwt2022simulated} produces the upper bound via simulated annealing. In this paper, we present our contribution to this problem in Section~\ref{rw-heur-section}. In particular, we construct a rank-decomposition from gflow with an explicit rank-width upper bound and introduce several heuristics that produce good rank-decompositions in practice.

\section{Constructing rank-decompositions}
\label{rw-heur-section}

In this section, we present our methods for generating rank-decompositions of a graph $G$. Refer to Appendix~\ref{rw-heur-appendix} for more details regarding proofs, implementation, and complexity analysis. Our first method, ``rw-flow'', generates a linear rank-decomposition from the extended gflow with an explicit rank-width upper bound. This upper bound helps us to show that our simulation algorithm, presented in Section~\ref{sim-section}, matches the scaling of the na\"ive state-vector simulation. We proceed to describe other methods that do not yet have rigorous performance guarantees. The first heuristic, ``rw-greedy-linear'', generates a linear rank-decomposition by greedily appending vertices to it. The second heuristic, ``rw-greedy-b2t'', on the other hand, produces highly nonlinear structures by greedily constructing the decomposition tree from the bottom up.

\subsection{rw-flow}
\label{rw-flow-subsection}

In the context of this method only, we assume that $G$ is an open graph with $|I| = |O| = n$. Our method starts by finding an extended gflow of $G$, that is, a tuple $(g, \prec, \lambda)$ where $\prec$ is the partial order and $\lambda(v) \in \{\mathrm{XY}, \mathrm{XZ}, \mathrm{YZ}\}$. This step takes polynomial time by \cite[Theorem~3.11]{backens2021there}. Next, we compute a linear extension of $\prec$, i.e., a permutation of vertices $(p_i)_{i = 1, \dots, |V(G)|}$ such that for each $i < j$ we have $p_i \not\succ p_j$. We can do this by topological sorting of the directed graph induced by $\prec$. Then, $p$ constitutes the desired linear rank-decomposition of width at most $n$ according to Lemma~\ref{gflow-rw-lemma} below. The proof of Lemma~\ref{gflow-rw-lemma} is given in Appendix~\ref{rw-flow-appendix}.
\begin{lemma}
    \label{gflow-rw-lemma}
    Let $(G, I, O)$ be an open graph, $(g, \prec, \lambda)$ be its extended gflow, and let $(p_i)_{i = 1, \dots, |V(G)|}$ be a permutation of vertices satisfying $p_i \not\succ p_j$ for each $i < j$. Then, $p$ constitutes a linear rank-decomposition of $G$, whose width does not exceed $|O|$.
\end{lemma}

\subsection{rw-greedy-linear}
\label{rw-greedy-linear-subsection}

Our greedy method for generating linear rank-decompositions is inspired by the Goemans-Soto technique for minimising symmetric submodular functions \cite{goemans2013algorithms}, based on the earlier Queyranne's work \cite{queyranne1998minimizing}. These works establish that a certain greedy algorithm correctly minimises a symmetric submodular function defined on subsets. Our method is a simplified version of their algorithm applied to the function $f \colon S \subset V \mapsto \rho_G(S) - \lambda |S|$ for some $\lambda \in \mathbb{R}$, which is submodular by \cite[Eq.~2]{oum2017rw}.

Given a simple undirected graph $G = (V, E)$, our heuristic generates a linear rank-decomposition $(p_i)_{i = 1, \dots, |V|}$ by iteratively appending vertices to it. We start from an empty $p$. For each $i = 1, \dots, |V|$ we pick $p_i$ using the following rule. If the biadjacency matrix $M$ between $S := \{p_1, \dots, p_{i-1}\}$ and $V \setminus S$ is zero, then we set $p_i$ to an arbitrary vertex from $V \setminus S$. Otherwise, we iterate over all pivot columns $u$ of $M$ and set $p_i := \mathrm{argmin}_u \,\rho_G(S \cup \{u\})$. Recall that a pivot column of $M$ is a column such that in the row-echelon form of $M$, it contains the leftmost non-zero entry in some row. Refer to Algorithm~\ref{naive-rw-greedy-linear-algorithm} for the pseudocode of this heuristic.

The choice of limiting our search to the pivot columns of $M$ was not arbitrary. As verified by our experiments on numerous examples, it not only reduces the complexity of the algorithm but also significantly improves the quality of the rank-decomposition. A rigorous theoretical explanation for this phenomenon remains open.

\subsection{rw-greedy-b2t}
\label{rw-greedy-b2t-subsection}

Another heuristic is inspired by the ``Hyper-Greedy'' contraction tree optimiser \cite{gray2021hyper}. In contrast to previous methods, it produces generic rank-decompositions by constructing the decomposition tree from its leaves to the root. We maintain a forest of partial decompositions $T = (V_T, E_T)$ and a partition of graph vertices $\{S_i\}_{i \in \mathrm{roots}(T)}$ corresponding to the set of leaves of each tree in $T$. We start from a collection of subtrees having a single leaf: $V_T = V$, $E_T = \varnothing$, and $S_v = \{v\}$ for each $v \in V$. Our method iteratively merges two trees in $T$ by adding a new vertex to $T$ as their common ancestor. As in the previous case, we consider a subset of \emph{valid} root pairs as candidates for this operation. Namely, a pair $(i, j)$ of roots in $T$ is considered valid if $S_i$ intersects the pivot columns of the biadjacency matrix between $S_j$ and $V \setminus S_j$. On each iteration, pick a valid pair $(i, j)$ minimising $\rho_G(S_i \cup S_j)$ (if there are no valid pairs, pick an arbitrary one). Merge $S_i$ and $S_j$ by adding to $T$ a new vertex $k$ with children $i, j$, setting $S_k := S_i \cup S_j$, and discarding $S_i$, $S_j$ (since $i, j$ are no longer roots). Repeat this procedure $|V| - 1$ times, so that only one tree remains, corresponding to the desired rank-decomposition tree. Refer to Algorithm~\ref{naive-rw-greedy-b2t-algorithm} for the pseudocode of this algorithm.

\section{Classical simulation}
\label{sim-section}

We focus on computing a single measurement amplitude $\langle y |C| x \rangle$ for a quantum circuit $C$, which is an important subroutine in most classical simulation algorithms. Solving this problem suffices to perform \textit{strong simulation} of quantum circuits, that is, computing arbitrary measurement probabilities and marginal probabilities (the latter by means of norm computations, a.k.a. ``doubling''). We can also use amplitude computations to perform weak simulation, i.e. sampling individual measurement outcomes, via techniques such as gate-by-gate circuit sampling~\cite{bravyi2022simulate}.

Generally, computing measurement amplitudes in polynomial time is infeasible, as it is just as hard as contracting an arbitrary tensor network, which is known to be \#P-complete~\cite{movassagh2023hardness}. In this section, we present an algorithm for contracting ZX-diagrams with complexity $\tilde O(4^R)$, where $R$ is the width of a rank-decomposition of the reduced ZX-diagram. Additionally, if $R$ is the width of a linear rank-decomposition, then our algorithm runs in $\tilde O(2^R)$ time (Corollary~\ref{linear-decomp-corollary}). In the case of quantum circuits, our method's complexity is bounded by $\tilde O(2^n)$ for a certain rank-decomposition (Corollary~\ref{naive-sim-corollary}). Furthermore, for circuits consisting of $t$ non-Clifford phase gates, we can construct a rank-decomposition that yields the simulation complexity of $\tilde O(2^{t / 2})$ (Corollary~\ref{t-count-corollary}). Our method is summarised by Algorithm~\ref{sim-algorithm}.

\begin{algorithm}
    \caption{Rank-width--based classical simulation of a quantum circuit $C$}
    \begin{enumerate}
        \item Convert $C$ into a ZX-diagram and simplify it as discussed in Section~\ref{zxcalc-subsection}, to obtain a reduced ZX-diagram $D$.
        \item Construct a rank-decomposition $T$ of $D$.
        \item Run the tensor contraction routine from Proposition~\ref{sim-proposition} along $T$.
    \end{enumerate}
    \label{sim-algorithm}
\end{algorithm}

The first step is standard for ZX-based simulation methods (e.g. stabiliser decompositions~\cite{kissinger2022simulating,kissinger2022classical,philipps2025gnns,sutcliffe2025novel}). We assume the circuit's inputs and outputs are plugged in so that $D$ is closed, except when we run the ``rw-flow'' routine in step~2, which works only for open graphs (as the extended gflow does not exist for closed ZX-diagrams). In practice, we may use all the techniques from Section~\ref{rw-heur-section} and choose the best rank-decomposition with respect to $\mathrm{flops}(T)$, defined below, which is a measure similar to rank-width characterising the cost of our tensor contraction routine.

\subsection{Tensor contraction routine}
\label{main-algo-subsection}

\begin{proposition}
    \label{sim-proposition}
    Let $D = (G, \alpha)$ be a graph-like ZX-diagram with no inputs and outputs, where $G = (V, E)$ is an undirected graph and $\alpha_v \in [0; 2\pi)$ denotes the phase of the vertex $v$. Let $T = (V_T, E_T)$ be a rank-decomposition of $G$ with cut-ranks $\rho_{u, v}$. Then, there exists an algorithm for contracting $D$ with complexity $\tilde O(\mathrm{flops}(T))$, where
    \begin{align*}
        \mathrm{flops}(T) &:= \sum_{v \in V_T} 2^{w_v}, & w_v &:=
        \begin{cases}
            \min(\rho_{v, u_1} + \rho_{v, u_2}, \rho_{v, u_1} + \rho_{v, u_3}, \rho_{v, u_2} + \rho_{v, u_3}) & \text{for $N_T(v) = \{u_1, u_2, u_3\}$}, \\
            0 & \text{for $|N_T(v)| = 1$}.
        \end{cases}
    \end{align*}
\end{proposition}

\begin{proof}
    First, we \emph{root} $T$ by deleting an arbitrary edge $(u_0, v_0) \in E_T$, orienting all remaining edges towards $u_0$ and $v_0$, adding a dummy vertex $0$ as a parent of $u_0$ and $v_0$, and setting the cut-ranks $\rho_{u_0, 0} = \rho_{v_0, 0} := \rho_{u_0, v_0}$. For each vertex $u$ in the rooted decomposition tree, consider the set of leaves $S_u$ from its subtree. Denote $r_u := \rho_G(S_u)$. Note that for each $u \neq 0$ we have $r_u = \rho_{u, p}$, where $p$ is the parent of $u$, and also $r_0 = 0$.

    Our contraction routine is a recursive procedure that takes a vertex $u$ of $T$ and ``simulates'' the subdiagram $D|_{S_u}$. Namely, it returns the state $\Psi_u$ and the parity matrix $M_u$ satisfying the equation in Figure~\ref{sim-conventions-figure}. The state $\Psi_u$ could be thought of as an ``encoded version'' of the subdiagram $D|_{S_u}$, which, in composition with the ``decoder'' $M_u$, would give the actual graph-like state induced by the spiders and edges of $D$ in $S_u$ (which has $|S_u|$ outputs). Note that our algorithm below never computes this state, as it would take $2^{\Omega(|V|)}$ time; instead, we work exclusively with the encoded versions, which takes only $2^{O(R)}$ time. To contract $D$, we call our procedure for $u = 0$ and return $\Psi_0$, which is a scalar since $r_0 = 0$.

    \begin{figure}
        \centering
        \begin{minipage}{0.55\textwidth}
            \centering
            \begin{tikzpicture}
	\begin{pgfonlayer}{nodelayer}
		\node [style=Z phase dot] (0) at (-2, 1.5) {$\alpha_1$};
		\node [style=Z phase dot] (1) at (-2, -1.5) {$\alpha_{|S_u|}$};
		\node [style=none] (2) at (-2, -0.25) {$\vdots$};
		\node [style=Z phase dot] (3) at (-2, 0.5) {$\alpha_2$};
		\node [style=none] (6) at (-1, 1.75) {};
		\node [style=none] (7) at (-1, 1.25) {};
		\node [style=none] (8) at (-1, 0.75) {};
		\node [style=none] (9) at (-1, 0.25) {};
		\node [style=none] (10) at (-1, -1.25) {};
		\node [style=none] (11) at (-1, -1.75) {};
		\node [style=none] (12) at (-1, -0.75) {};
		\node [style=none] (13) at (-1, -0.25) {};
		\node [style=none] (14) at (-2.25, -0.75) {};
		\node [style=none] (15) at (-2.25, -0.25) {};
		\node [style=none] (18) at (0, 0) {$=$};
		\node [style=none] (20) at (1.5, 0) {$r_u \Biggl\{$};
		\node [style=ket] (21) at (3.5, 0) {$\, \Psi_u \,$};
		\node [style=none] (22) at (4, 1) {};
		\node [style=none] (23) at (4, -1) {};
		\node [style=Z dot] (24) at (5, 1) {};
		\node [style=Z dot] (4) at (5, -1) {};
		\node [style=none] (25) at (4.75, 0.25) {$\vdots$};
		\node [style=Z dot] (26) at (7, 1.5) {};
		\node [style=Z dot] (27) at (7, 0.5) {};
		\node [style=Z dot] (28) at (7, -1.5) {};
		\node [style=Z dot] (29) at (8, 1.5) {};
		\node [style=Z dot] (30) at (8, 0.5) {};
		\node [style=Z dot] (31) at (8, -1.5) {};
		\node [style=none] (32) at (7.5, -0.25) {$\vdots$};
		\node [style=none] (33) at (9, 1.75) {};
		\node [style=none] (34) at (9, 1.25) {};
		\node [style=none] (35) at (9, 0.75) {};
		\node [style=none] (36) at (9, 0.25) {};
		\node [style=none] (37) at (9, -1.25) {};
		\node [style=none] (38) at (9, -1.75) {};
		\node [style=none] (39) at (9, -0.75) {};
		\node [style=none] (40) at (9, -0.25) {};
		\node [style=none] (42) at (6, 2.25) {$M_u$};
		\node [style=none] (43) at (-7.5, -2.25) {};
		\node [style=none] (44) at (-7.5, 2.5) {};
		\node [style=none] (45) at (9.5, -2.25) {};
		\node [style=none] (46) at (9.5, 2.5) {};
		\node [style=none] (47) at (-3.5, 0) {$=$};
		\node [style=none] (51) at (-4.5, 1) {};
		\node [style=none] (52) at (-4.5, -1) {};
		\node [style=none] (55) at (-4.75, 0.25) {$\vdots$};
		\node [style=vertex set] (58) at (-6.25, 0) {$\quad S_u \quad$};
	\end{pgfonlayer}
	\begin{pgfonlayer}{edgelayer}
		\draw [style=hadamard edge, bend right=15] (0) to (3);
		\draw [style=hadamard edge, bend right=45] (0) to (1);
		\draw [style=hadamard edge] (0) to (6.center);
		\draw [style=hadamard edge] (0) to (7.center);
		\draw [style=hadamard edge] (3) to (8.center);
		\draw [style=hadamard edge] (3) to (9.center);
		\draw [style=hadamard edge] (1) to (10.center);
		\draw [style=hadamard edge] (1) to (11.center);
		\draw [style=hadamard edge] (1) to (12.center);
		\draw [style=hadamard edge] (3) to (13.center);
		\draw [style=hadamard edge, bend left=15] (1) to (14.center);
		\draw [style=hadamard edge, bend right=15] (3) to (15.center);
		\draw (22.center) to (24);
		\draw (23.center) to (4);
		\draw [style=hadamard edge] (24) to (26);
		\draw [style=hadamard edge] (24) to (27);
		\draw [style=hadamard edge] (4) to (27);
		\draw [style=hadamard edge] (4) to (28);
		\draw [style=hadamard edge] (26) to (29);
		\draw [style=hadamard edge] (27) to (30);
		\draw [style=hadamard edge] (28) to (31);
		\draw [style=hadamard edge] (24) to (28);
		\draw [style=hadamard edge] (29) to (33.center);
		\draw [style=hadamard edge] (29) to (34.center);
		\draw [style=hadamard edge] (30) to (35.center);
		\draw [style=hadamard edge] (30) to (36.center);
		\draw [style=hadamard edge] (30) to (40.center);
		\draw [style=hadamard edge] (31) to (38.center);
		\draw [style=hadamard edge] (31) to (37.center);
		\draw [style=hadamard edge] (31) to (39.center);
		\draw [style=hadamard edge, bend left=15] (58) to (51.center);
		\draw [style=hadamard edge, bend right=15] (58) to (52.center);
	\end{pgfonlayer}
\end{tikzpicture}
            \caption{Conventions for our simulation routine. $\Psi_u$ is a $r_u$-qubit state, and $M_u$ is a parity matrix of size $r_u \times |S_u|$. The rightmost edges (between $S_u$ and $V \setminus S_u$) are unchanged.}
            \label{sim-conventions-figure}
        \end{minipage}
        \hspace{0.05\textwidth}
        \begin{minipage}{0.35\textwidth}
            \centering
            \begin{tikzpicture}
	\begin{pgfonlayer}{nodelayer}
		\node [style=Z phase dot] (0) at (-2.5, 0) {$\alpha_u$};
		\node [style=none] (1) at (-1, 1.5) {};
		\node [style=none] (2) at (-1, 0.75) {};
		\node [style=none] (3) at (-1, -1.5) {};
		\node [style=none] (5) at (-1.25, -0.25) {$\vdots$};
		\node [style=none] (6) at (0, 0) {$=$};
		\node [style=Z phase dot] (7) at (1.5, 0) {$\alpha_u$};
		\node [style=Z dot] (8) at (2.5, 0) {};
		\node [style=Z dot] (9) at (3.5, 0) {};
		\node [style=Z dot] (10) at (4.5, 0) {};
		\node [style=none] (11) at (6, 1.5) {};
		\node [style=none] (12) at (6, 0.75) {};
		\node [style=none] (13) at (6, -1.5) {};
		\node [style=none] (14) at (5.75, -0.25) {$\vdots$};
		\node [style=none] (15) at (-3, -2) {};
		\node [style=none] (16) at (-3, 2) {};
		\node [style=none] (17) at (6.5, -2) {};
		\node [style=none] (18) at (6.5, 2) {};
		\node [style=none] (19) at (1.5, 1) {$\Psi_u$};
		\node [style=none] (20) at (3, 1) {$M_u$};
	\end{pgfonlayer}
	\begin{pgfonlayer}{edgelayer}
		\draw [style=hadamard edge, bend left] (0) to (1.center);
		\draw [style=hadamard edge, bend left=15] (0) to (2.center);
		\draw [style=hadamard edge, bend right] (0) to (3.center);
		\draw (7) to (8);
		\draw [style=hadamard edge] (8) to (9);
		\draw [style=hadamard edge] (9) to (10);
		\draw [style=hadamard edge, bend left] (10) to (11.center);
		\draw [style=hadamard edge, bend left=15] (10) to (12.center);
		\draw [style=hadamard edge, bend right] (10) to (13.center);
	\end{pgfonlayer}
\end{tikzpicture}
            \caption{Simulation of $D|_{S_u}$: leaf case}
            \label{sim-leaf-figure}
        \end{minipage}
    \end{figure}

    Our recursive procedure works as follows. If $u$ is a leaf of $T$, then $\Psi_u$ is just a single-legged spider and $M_u$ is an identity matrix -- see Figure~\ref{sim-leaf-figure}. Otherwise, we follow the course of reasoning presented in Figure~\ref{sim-general-figure}:

    \begin{enumerate}[label = \textbf{Step \arabic{*}.}, align = right, leftmargin = 4em]
        \item Denote by $v, w$ the children of $u$ in $T$, and let $A, B$ be the biadjacency matrices defined by the pairs of vertex sets $(S_u, V \setminus S_u)$ and $(S_v, S_w)$, respectively. Note that $S_u = S_v \sqcup S_w$.
        \item Recursively simulate $D|_{S_v}, D|_{S_w}$ to obtain $\Psi_v, \Psi_w, M_v, M_w$.
        \item $B$ now corresponds to a layer of CZ gates after the parity map defined by $M_v \oplus M_w$. We can shift this to a layer of CZ gates before the parity map given by $E_{vw} := M_v B M_w^T$.
        \item Find $|S_u| \times r_u, r_u \times |S_u|$ matrices $U, V$ such that $A = UVA$. This is possible because $A$ has rank $r_u$, as guaranteed by the rank-decomposition. That is, we can write $A = PA$, where $P$ is the projector onto the image of $A$. We can then split the projector $P$ as a pair of matrices $U, V$ such that $P = UV$ and $VU = I_{r_u}$.
        \item Let $U_v$ be the first $|S_v|$ rows of $U$, and let $U_w$ be the last $|S_w|$ rows. Composing the parity maps again, compute $E_{vu} := M_v U_v$ and $E_{wu} := M_w U_w$.
        \item Finally, perform the \emph{convolution}, further detailed below, to contract the left half of the diagram and obtain $\Psi_u$. Setting $M_u := V$, we obtain $\Psi_u, M_u$ in the required form for the recursion at \textbf{Step 2}. If $u$ is the root, then $\Psi_u$ is just a scalar, so we return it as the result.
    \end{enumerate}

    \begin{figure}
        \centering
        \input{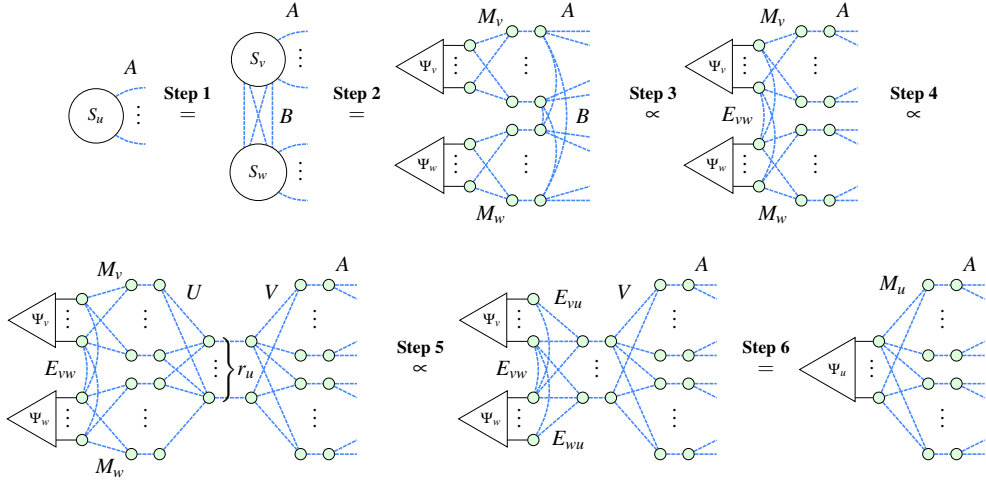}
        \caption{Simulation of $D|_{S_u}$: general case}
        \label{sim-general-figure}
    \end{figure}

    The convolution subroutine contracts the tripartite ZX-diagram as shown in Eq.~\eqref{conv-equation}. Recall that $\Psi_v, \Psi_w$ are $r_v$- and $r_w$-qubit states, respectively, and the output is a $r_u$-qubit state $\widehat \Psi_u$ (which relates to $\Psi_u$ via Hadamard transform). The parity matrices $E_{vu}, E_{wu}, E_{vw}$ denote the edges between the three parts. We show that the convolution can be performed in $\tilde O\left(2^{\min(r_v + r_w, r_u + r_v, r_u + r_w)}\right)$ time using the cheapest of three equivalent methods described below. Since steps 1, 3, 4, and 5 take polynomial time, by induction over $u$ it will follow that the overall complexity of our simulation routine is $\tilde O(\mathrm{flops}(T))$.

    \begin{align}
        \label{conv-equation}
        \begin{tikzpicture}
	\begin{pgfonlayer}{nodelayer}
		\node [style=ket] (0) at (0, 1.75) {$\,\Psi_v$};
		\node [style=ket] (1) at (0, -1.75) {$\Psi_w$};
		\node [style=none] (2) at (0.5, 2.5) {};
		\node [style=none] (3) at (0.5, 1) {};
		\node [style=none] (4) at (0.5, -1) {};
		\node [style=none] (5) at (0.5, -2.5) {};
		\node [style=Z dot] (6) at (1.5, 2.5) {};
		\node [style=Z dot] (7) at (1.5, 1) {};
		\node [style=Z dot] (8) at (1.5, -1) {};
		\node [style=Z dot] (9) at (1.5, -2.5) {};
		\node [style=Z dot] (10) at (4, 1) {};
		\node [style=Z dot] (11) at (4, -1) {};
		\node [style=none] (12) at (4.5, 1) {};
		\node [style=none] (13) at (4.5, -1) {};
		\node [style=none] (14) at (1, 2) {$\vdots$};
		\node [style=none] (15) at (1, -1.5) {$\vdots$};
		\node [style=none] (16) at (4.25, 0.25) {$\vdots$};
		\node [style=none] (17) at (3, 2.5) {$E_{vu}$};
		\node [style=none] (18) at (3, -2.5) {$E_{wu}$};
		\node [style=none] (19) at (1, 0) {$E_{vw}$};
	\end{pgfonlayer}
	\begin{pgfonlayer}{edgelayer}
		\draw (2.center) to (6);
		\draw (3.center) to (7);
		\draw (4.center) to (8);
		\draw (5.center) to (9);
		\draw (10) to (12.center);
		\draw (11) to (13.center);
		\draw [style=hadamard edge] (6) to (10);
		\draw [style=hadamard edge] (6) to (11);
		\draw [style=hadamard edge] (7) to (11);
		\draw [style=hadamard edge] (8) to (10);
		\draw [style=hadamard edge] (8) to (11);
		\draw [style=hadamard edge] (9) to (11);
		\draw [style=hadamard edge, bend left] (6) to (8);
		\draw [style=hadamard edge, bend left] (7) to (9);
	\end{pgfonlayer}
\end{tikzpicture} \weq \begin{tikzpicture}
	\begin{pgfonlayer}{nodelayer}
		\node [style=ket] (0) at (0, 0) {$\,\,\widehat\Psi_u\,$};
		\node [style=none] (1) at (0.5, 1) {};
		\node [style=none] (2) at (0.5, -1) {};
		\node [style=none] (3) at (1.5, 1) {};
		\node [style=none] (4) at (1.5, -1) {};
		\node [style=none] (5) at (1.25, 0.25) {$\vdots$};
	\end{pgfonlayer}
	\begin{pgfonlayer}{edgelayer}
		\draw (1.center) to (3.center);
		\draw (2.center) to (4.center);
	\end{pgfonlayer}
\end{tikzpicture}
    \end{align}
    \vspace{5pt}

    The first approach for the convolution is summarised in Figure~\ref{conv-vw-figure}. We start by computing $\Psi_v \otimes \Psi_w$. Next, we apply the diagonal CZ-only operator corresponding to the $E_{vw}$-edges. Finally, we apply the parity map $(E_{vu}, E_{wu})$. Note that all these steps take $\tilde O(2^{r_v + r_w})$ time, since the inequality $r_u \le r_v + r_w$ follows from the submodular inequalities for cut-ranks~\cite[Eq.~2]{oum2017rw}.

    The second method, summarised in Figure~\ref{conv-uv-figure}, first rotates the $\Psi_v$ state (so that it becomes a post-selection) and then unfuses the $E_{vu}$-edges. We start by applying the parity map $(E_{vw}^T, E_{wu})$ to $\Psi_w$. Next, we apply the diagonal CZ-only operator corresponding to the $E_{vu}$-edges. Finally, we perform the post-selection with $\Psi_v$. Note that all these steps take $\tilde O(2^{r_u + r_v})$ time, since $r_w \le r_u + r_v$ by \cite[Eq.~2]{oum2017rw}.

    The third method with complexity $\tilde O(2^{r_u + r_w})$ is just the second one with $\Psi_v$ and $\Psi_w$ interchanged.

    \begin{figure}[htbp]
        \centering
        \begin{minipage}{0.49\textwidth}
            \centering
            \begin{tikzpicture}
	\begin{pgfonlayer}{nodelayer}
		\node [style=ket] (0) at (-5.5, 1.75) {$\,\Psi_v$};
		\node [style=ket] (1) at (-5.5, -1.75) {$\Psi_w$};
		\node [style=none] (2) at (-5, 2.5) {};
		\node [style=none] (3) at (-5, 1) {};
		\node [style=none] (4) at (-5, -1) {};
		\node [style=none] (5) at (-5, -2.5) {};
		\node [style=Z dot] (6) at (-4, 2.5) {};
		\node [style=Z dot] (7) at (-4, 1) {};
		\node [style=Z dot] (8) at (-4, -1) {};
		\node [style=Z dot] (9) at (-4, -2.5) {};
		\node [style=Z dot] (10) at (-1.5, 1) {};
		\node [style=Z dot] (11) at (-1.5, -1) {};
		\node [style=none] (12) at (-1, 1) {};
		\node [style=none] (13) at (-1, -1) {};
		\node [style=none] (14) at (-4.5, 2) {$\vdots$};
		\node [style=none] (15) at (-4.5, -1.5) {$\vdots$};
		\node [style=none] (16) at (-1.25, 0.25) {$\vdots$};
		\node [style=none] (17) at (-2.5, 2.5) {$E_{vu}$};
		\node [style=none] (18) at (-2.5, -2.5) {$E_{wu}$};
		\node [style=none] (19) at (-4.75, 0) {$E_{vw}$};
		\node [style=none] (20) at (0, 0) {$=$};
		\node [style=ket] (21) at (2, 1.75) {$\,\Psi_v$};
		\node [style=ket] (22) at (2, -1.75) {$\Psi_w$};
		\node [style=none] (23) at (2.5, 2.5) {};
		\node [style=none] (24) at (2.5, 1) {};
		\node [style=none] (25) at (2.5, -1) {};
		\node [style=none] (26) at (2.5, -2.5) {};
		\node [style=Z dot] (27) at (4.5, 2.5) {};
		\node [style=Z dot] (28) at (4.5, 1) {};
		\node [style=Z dot] (29) at (4.5, -1) {};
		\node [style=Z dot] (30) at (4.5, -2.5) {};
		\node [style=Z dot] (31) at (7, 1) {};
		\node [style=Z dot] (32) at (7, -1) {};
		\node [style=none] (33) at (7.5, 1) {};
		\node [style=none] (34) at (7.5, -1) {};
		\node [style=none] (35) at (3, 2) {$\vdots$};
		\node [style=none] (36) at (3, -1.5) {$\vdots$};
		\node [style=none] (37) at (7.25, 0.25) {$\vdots$};
		\node [style=Z dot] (38) at (3.5, 2.5) {};
		\node [style=Z dot] (39) at (3.5, 1) {};
		\node [style=Z dot] (40) at (3.5, -1) {};
		\node [style=Z dot] (41) at (3.5, -2.5) {};
		\node [style=none] (42) at (4, -1.5) {$\vdots$};
		\node [style=none] (43) at (4, 2) {$\vdots$};
		\node [style=none] (44) at (6, 2.5) {$E_{vu}$};
		\node [style=none] (45) at (6, -2.5) {$E_{wu}$};
		\node [style=none] (46) at (2.75, 0) {$E_{vw}$};
	\end{pgfonlayer}
	\begin{pgfonlayer}{edgelayer}
		\draw (2.center) to (6);
		\draw (3.center) to (7);
		\draw (4.center) to (8);
		\draw (5.center) to (9);
		\draw (10) to (12.center);
		\draw (11) to (13.center);
		\draw [style=hadamard edge] (6) to (10);
		\draw [style=hadamard edge] (6) to (11);
		\draw [style=hadamard edge] (7) to (11);
		\draw [style=hadamard edge] (8) to (10);
		\draw [style=hadamard edge] (8) to (11);
		\draw [style=hadamard edge] (9) to (11);
		\draw [style=hadamard edge, bend left] (6) to (8);
		\draw [style=hadamard edge, bend left] (7) to (9);
		\draw (23.center) to (27);
		\draw (24.center) to (28);
		\draw (25.center) to (29);
		\draw (26.center) to (30);
		\draw (31) to (33.center);
		\draw (32) to (34.center);
		\draw [style=hadamard edge] (27) to (31);
		\draw [style=hadamard edge] (27) to (32);
		\draw [style=hadamard edge] (28) to (32);
		\draw [style=hadamard edge] (29) to (31);
		\draw [style=hadamard edge] (29) to (32);
		\draw [style=hadamard edge] (30) to (32);
		\draw [style=hadamard edge, bend left=15] (38) to (40);
		\draw [style=hadamard edge, bend left=15] (39) to (41);
	\end{pgfonlayer}
\end{tikzpicture}
            \caption{Convolution with complexity $\tilde O(2^{r_v + r_w})$}
            \label{conv-vw-figure}
        \end{minipage}
        \begin{minipage}{0.49\textwidth}
            \centering
            \begin{tikzpicture}
	\begin{pgfonlayer}{nodelayer}
		\node [style=ket] (0) at (-5.5, 1.75) {$\,\Psi_v$};
		\node [style=ket] (1) at (-5.5, -1.75) {$\Psi_w$};
		\node [style=none] (2) at (-5, 2.5) {};
		\node [style=none] (3) at (-5, 1) {};
		\node [style=none] (4) at (-5, -1) {};
		\node [style=none] (5) at (-5, -2.5) {};
		\node [style=Z dot] (6) at (-4, 2.5) {};
		\node [style=Z dot] (7) at (-4, 1) {};
		\node [style=Z dot] (8) at (-4, -1) {};
		\node [style=Z dot] (9) at (-4, -2.5) {};
		\node [style=Z dot] (10) at (-1.5, 1) {};
		\node [style=Z dot] (11) at (-1.5, -1) {};
		\node [style=none] (12) at (-1, 1) {};
		\node [style=none] (13) at (-1, -1) {};
		\node [style=none] (14) at (-4.5, 2) {$\vdots$};
		\node [style=none] (15) at (-4.5, -1.5) {$\vdots$};
		\node [style=none] (16) at (-1.25, 0.25) {$\vdots$};
		\node [style=none] (17) at (-2.5, 2.5) {$E_{vu}$};
		\node [style=none] (18) at (-2.5, -2.5) {$E_{wu}$};
		\node [style=none] (19) at (-4.75, 0) {$E_{vw}$};
		\node [style=none] (20) at (0, 0) {$=$};
		\node [style=ket] (21) at (2.25, 0) {$\Psi_w$};
		\node [style=none] (22) at (8.75, 2.5) {};
		\node [style=none] (23) at (8.75, 1) {};
		\node [style=none] (24) at (2.75, 0.75) {};
		\node [style=none] (25) at (2.75, -0.75) {};
		\node [style=Z dot] (26) at (6.25, 2.5) {};
		\node [style=Z dot] (27) at (6.25, 1) {};
		\node [style=Z dot] (28) at (3.75, 0.75) {};
		\node [style=Z dot] (29) at (3.75, -0.75) {};
		\node [style=Z dot] (30) at (6.25, -1) {};
		\node [style=Z dot] (31) at (6.25, -2.5) {};
		\node [style=none] (32) at (8.75, -1) {};
		\node [style=none] (33) at (8.75, -2.5) {};
		\node [style=none] (34) at (6.75, 2) {$\vdots$};
		\node [style=none] (35) at (3.25, 0.25) {$\vdots$};
		\node [style=none] (36) at (6.75, -1.5) {$\vdots$};
		\node [style=none] (37) at (9, 0) {$E_{vu}$};
		\node [style=none] (38) at (4.75, -2.5) {$E_{wu}$};
		\node [style=none] (39) at (4.75, 2.5) {$E_{vw}^T$};
		\node [style=bra] (40) at (9.25, 1.75) {$\,\Psi_v$};
		\node [style=Z dot] (41) at (7.25, 2.5) {};
		\node [style=Z dot] (42) at (7.25, 1) {};
		\node [style=Z dot] (43) at (7.25, -1) {};
		\node [style=Z dot] (44) at (7.25, -2.5) {};
		\node [style=none] (45) at (8.25, 2) {$\vdots$};
		\node [style=none] (46) at (8.25, -1.5) {$\vdots$};
	\end{pgfonlayer}
	\begin{pgfonlayer}{edgelayer}
		\draw (2.center) to (6);
		\draw (3.center) to (7);
		\draw (4.center) to (8);
		\draw (5.center) to (9);
		\draw (10) to (12.center);
		\draw (11) to (13.center);
		\draw [style=hadamard edge] (6) to (10);
		\draw [style=hadamard edge] (6) to (11);
		\draw [style=hadamard edge] (7) to (11);
		\draw [style=hadamard edge] (8) to (10);
		\draw [style=hadamard edge] (8) to (11);
		\draw [style=hadamard edge] (9) to (11);
		\draw [style=hadamard edge, bend left] (6) to (8);
		\draw [style=hadamard edge, bend left] (7) to (9);
		\draw (22.center) to (26);
		\draw (23.center) to (27);
		\draw (24.center) to (28);
		\draw (25.center) to (29);
		\draw (30) to (32.center);
		\draw (31) to (33.center);
		\draw [style=hadamard edge] (28) to (30);
		\draw [style=hadamard edge] (28) to (31);
		\draw [style=hadamard edge] (29) to (31);
		\draw [style=hadamard edge] (26) to (28);
		\draw [style=hadamard edge] (27) to (29);
		\draw [style=hadamard edge, bend left] (42) to (44);
		\draw [style=hadamard edge, bend left] (41) to (43);
		\draw [style=hadamard edge, bend left] (41) to (44);
	\end{pgfonlayer}
\end{tikzpicture}
            \caption{Convolution with complexity $\tilde O(2^{r_u + r_v})$}
            \label{conv-uv-figure}
        \end{minipage}
    \end{figure}
\end{proof}

\begin{corollary}
    \label{general-decomp-corollary}
    Under the assumptions from Proposition~\ref{sim-proposition}, the simulation complexity can be upper bounded by $\tilde O(4^R)$, where $R$ is the width of the rank-decomposition $T$.
\end{corollary}

\begin{proof}
    Note that we have $w_v \le 2R$ for each $v \in V_T$. Hence, $\mathrm{flops}(T) \le |V| \cdot 2^{2R} = \tilde O(4^R)$.
\end{proof}

\begin{corollary}
    \label{linear-decomp-corollary}
    Under the assumptions from Proposition~\ref{sim-proposition}, if additionally $T$ is linear, then the simulation complexity can be upper bounded by $\tilde O(2^R)$, where $R$ is the width of $T$.
\end{corollary}

\begin{proof}
    Since $T$ is linear, we have $w_v \le R + 1$. Hence, $\mathrm{flops}(T) \le |V| \cdot 2^{R + 1} = \tilde O(2^R)$.
\end{proof}

\begin{corollary}
    \label{naive-sim-corollary}
    Algorithm~\ref{sim-algorithm} simulates an $n$-qubit quantum circuit in $\tilde O(2^n)$ time for a certain efficiently constructible decomposition $T$.
\end{corollary}

\begin{proof}
    The reduced ZX-diagram of a quantum circuit is computable in polynomial time; see e.g. \cite[Chapter~7]{kissinger2024pqs}. Let $(p_i)_{i = 1, \dots, |V(D)|}$ denote the linear rank-decomposition generated by ``rw-flow'', which also takes polynomial time; see Section~\ref{rw-flow-subsection}. By Lemma~\ref{gflow-rw-lemma}, the width of $p$ does not exceed $n$. The complexity of Algorithm~\ref{sim-algorithm} is then dominated by the last step, which costs $\tilde O(2^n)$ by Corollary~\ref{linear-decomp-corollary}.
\end{proof}

\begin{corollary}
    \label{t-count-corollary}
    Let $C$ be a quantum circuit consisting of Clifford gates and $t$ non-Clifford phase gates. Then, Algorithm~\ref{sim-algorithm} simulates $C$ in $\tilde O(2^{t / 2})$ time for a certain efficiently constructible decomposition $T$.
\end{corollary}

\begin{proof}
    First, note that the ZX-diagram of $C$ has $t$ non-Clifford spiders. There are $t' \le t$ non-Clifford spiders in its reduced ZX-diagram $D$ because the ZX simplifications can only eliminate them. Since $D$ is closed, each Clifford spider is a part of a phase gadget. Consider an arbitrary linear rank-decomposition $(p_i)_{i = 1, \dots, t'}$ of the internal non-Clifford spiders and the phase gadgets (treated as single nodes). It extends to a (nonlinear) rank-decomposition $T$ of $D$ by adding two children to each phase gadget node. Note that the width of $(p_i)$ is at most $t' / 2$, since its cut-ranks $r_i$ satisfy $r_i \le \min(i, t' - i)$. Thus, the width of $T$ is also bounded by $t' / 2$. A direct estimate of $\mathrm{flops}(T)$, analogous to the proof of Corollary~\ref{linear-decomp-corollary}, shows that Algorithm~\ref{sim-algorithm} simulates $D$ in $\tilde O(2^{t' / 2})$ time.
\end{proof}

\section{Benchmarks}
\label{bench-section}

We benchmark our tensor contraction routine from Section~\ref{main-algo-subsection} with different rank-decomposition strategies from Section~\ref{rw-heur-section} on certain families of quantum circuits and ZX-diagrams. The pipeline is essentially the same as in Algorithm~\ref{sim-algorithm}, with one slight modification: instead of running the contraction routine directly, we estimate the number of \emph{flops} -- floating-point operations required for the contractions. This number is given by a formula similar to that of Proposition~\ref{sim-proposition}, with each term multiplied by a certain linear combination of the cut-ranks. We also estimate the number of flops for the na\"ive statevector-like simulation routine \texttt{tensorfy} provided by the PyZX library. As a baseline, we use the Quimb library, which contains analogous tensor contraction routines for hypergraphs, for which the number of flops can also be measured in advance. Similar to the ZX-based methods, Quimb first simplifies the given tensor network. Then it runs an optimiser to find an appropriate sequence of tensor contractions (aka a contraction tree). We use the high-quality ``auto-hq'' optimiser, which takes a longer time to converge on larger instances but aims to produce efficient contraction trees. Both \texttt{tensorfy} and Quimb routines are run on the original (non-reduced) ZX-diagram. We have also verified the correctness of all the compared routines on the instances below, where the number of flops is not too large.

\subsection{Random CNOT + H + T circuits}

We generate a family of random CNOT + H + T circuits by iteratively appending CNOT, Hadamard, and T gates with probabilities 0.6, 0.2, and 0.2, respectively. First, we fix the qubit count to 10 and vary the number of gates (Figure~\ref{bench-random-Q10-figure}). Then, we set the gate count to $3 \cdot \mathrm{n\_qubits}^2$ and vary the number of qubits (Figure~\ref{bench-random-Q2-15-figure}). This scaling was chosen to ensure circuits do not become too sparse as the qubit count is increased, trivialising the classical simulation task. For each configuration of $(\mathrm{n\_qubits}, \mathrm{n\_gates})$, the results were averaged over five independent runs.

\begin{figure}[htbp]
    \centering
    \begin{subfigure}[c]{0.4\textwidth}
        \includegraphics[width=\textwidth]{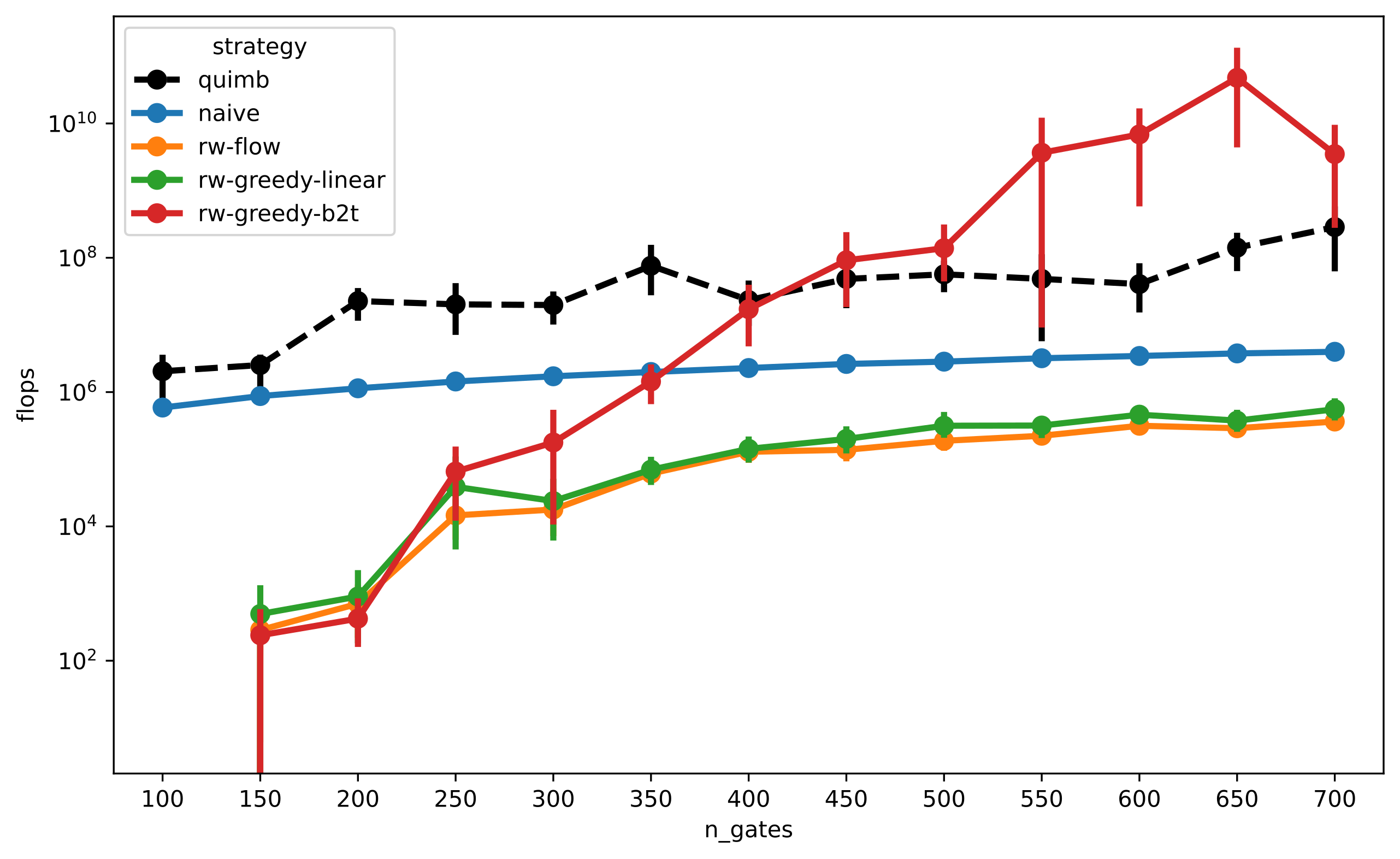}
        \caption{for $\mathrm{n\_qubits} = 10$}
        \label{bench-random-Q10-figure}
    \end{subfigure}
    \begin{subfigure}[c]{0.4\textwidth}
        \includegraphics[width=\textwidth]{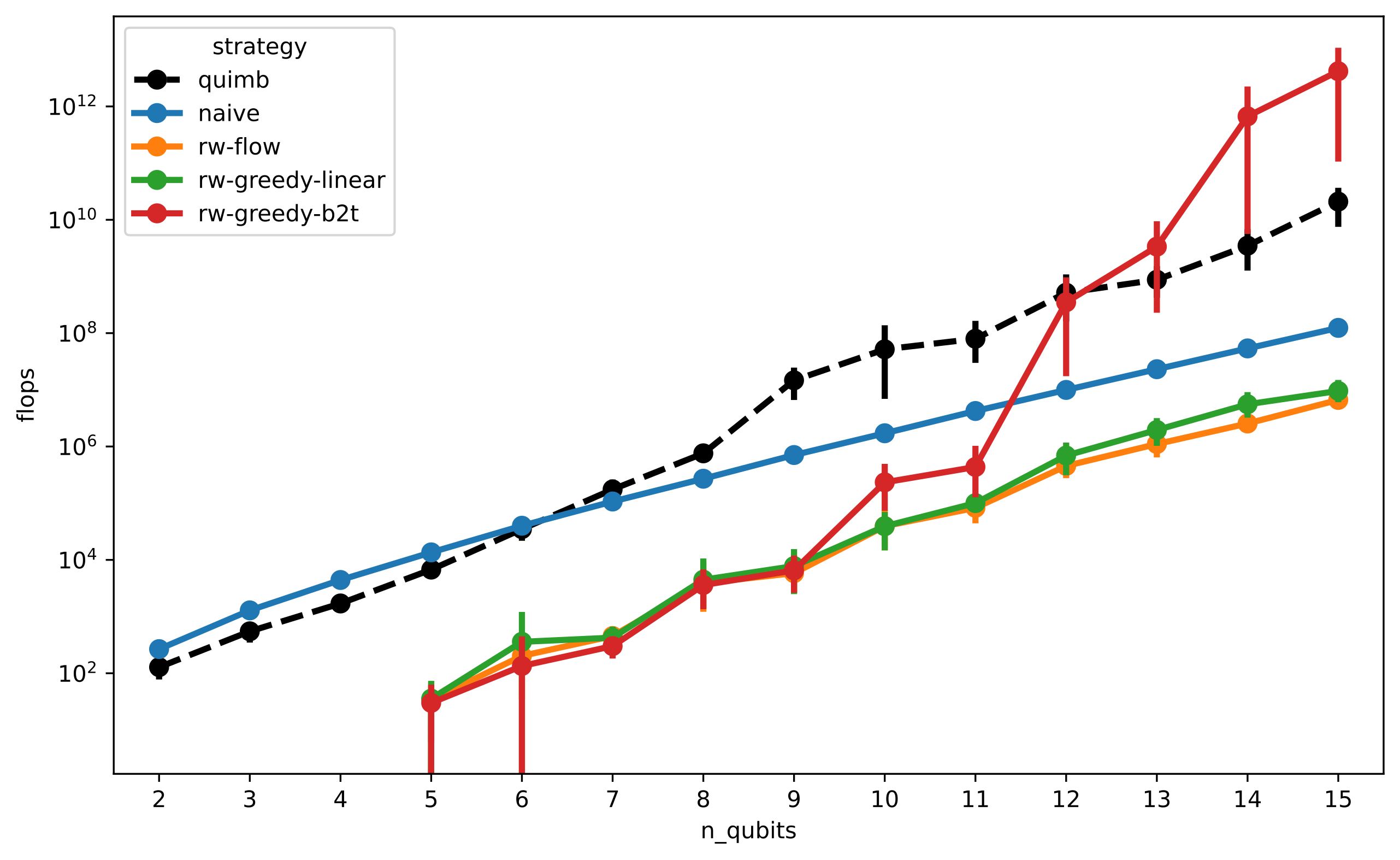}
        \caption{for $\mathrm{n\_gates} = 3 \cdot \mathrm{n\_qubits}^2$}
        \label{bench-random-Q2-15-figure}
    \end{subfigure}
    \caption{Benchmark on random CNOT + H + T circuits}
\end{figure}

When the number of qubits is fixed, the naive simulation appears to outperform Quimb by a constant factor of 10 on almost all instances. Our routines are very efficient for the shallow circuits since their T-count is small. For the deeper ones, ``rw-flow'' and ``rw-greedy-linear'' level off at the 100-times-faster mark compared to Quimb, but ``rw-greedy-b2t'' starts to slow down rapidly. This is likely due to the large cut-ranks appearing at the top of the rank-decomposition tree, since ``rw-greedy-b2t'' does not take them into account while building the decomposition from the bottom up.

When the number of qubits is varied, we observe that each simulation routine slows down exponentially. This is because for sufficiently dense random circuits, their rank-width appears to be close to the number of qubits. Thus, rank-width--based methods do not constitute a significant improvement over the naive method in this case.

\subsection{Structured circuits}

We run our routines on a set of structured circuits taken from the PyZX repository\footnote{\url{https://github.com/zxcalc/pyzx}}, which are mainly based on the T-Par~\cite{amy2014polynomial} benchmark circuits\footnote{\url{https://github.com/meamy/t-par}}. Since the majority of circuits represent a classical computation, we set the inputs and outputs of these circuits to the T-states. Figure~\ref{bench-structured-figure} illustrates the performance of our methods relative to the baseline.

On every circuit in the dataset, our methods in combination perform similarly to or faster than Quimb. In most cases, ``rw-greedy-b2t'' performs best, and in some cases, it is up to 10,000 times faster than Quimb. The ``rw-flow'' routine complements ``rw-greedy-b2t'': it is around the baseline when the greedy routine performs poorly, and is around 10 times slower than the baseline when the greedy routine dominates. This benchmark demonstrates that a good choice of a rank-decomposition can significantly impact the performance of our contraction routine, and the combination of ``rw-greedy-b2t'' and ``rw-flow'' is a robust choice in these benchmarks.

\begin{figure}[htbp]
    \centering
    \includegraphics[width=0.75\textwidth]{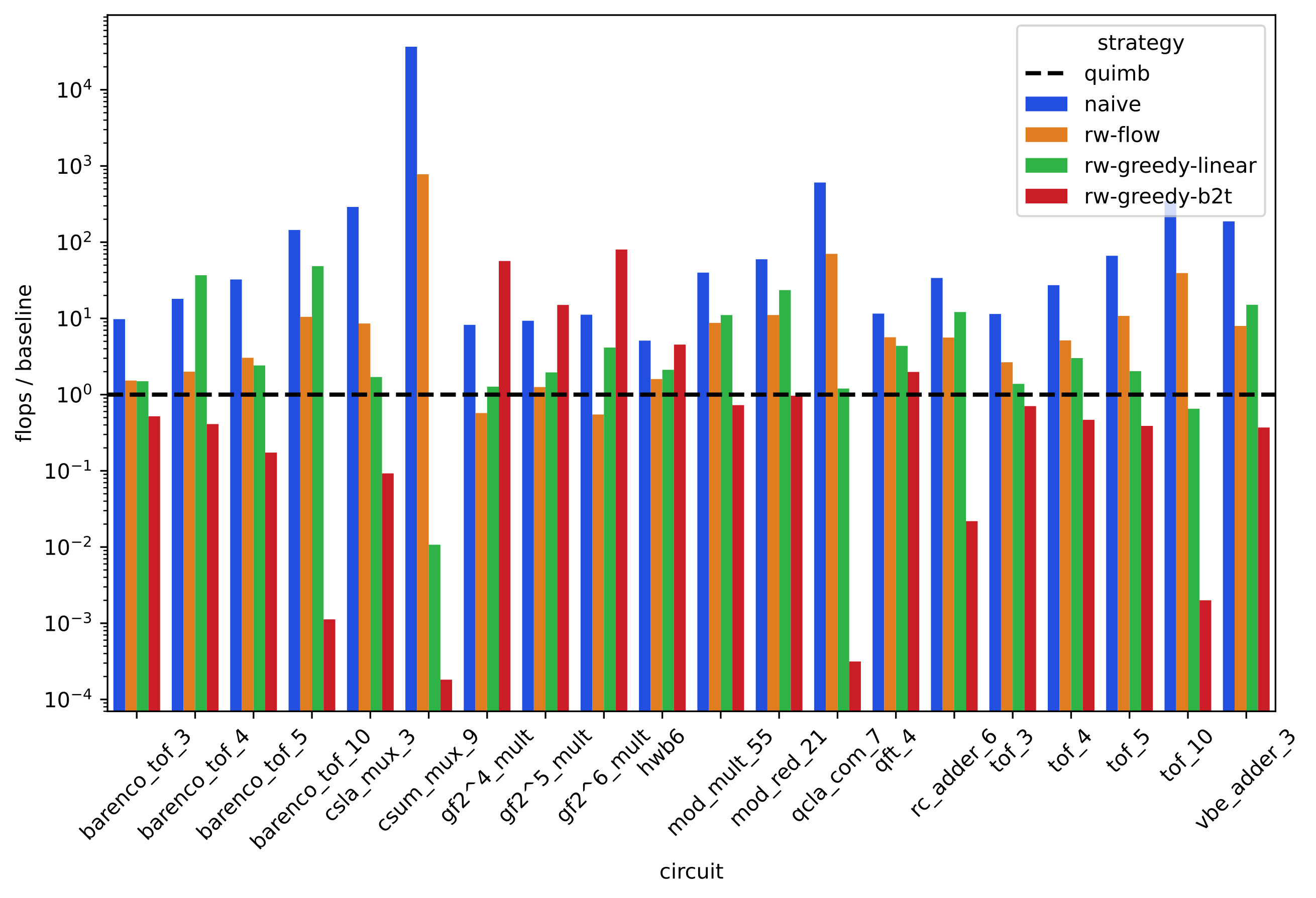}
    \caption{Benchmark on structured circuits from the PyZX repository}
    \label{bench-structured-figure}
\end{figure}

\subsection{Multi-qubit Toffoli}

The benchmark above motivates us to consider the family of circuits implementing an $n$-qubit Toffoli gate (also known as the multi-controlled NOT gate). It is well-known that these circuits have stabiliser rank 2 and are thus efficiently simulable. As before, we set the circuit inputs and outputs to T-states to avoid total simplification. The results are shown in Figure~\ref{bench-tof-figure}.

Notably, the performance of ``rw-greedy-b2t'' appears to scale polynomially with $n$, while other methods result in exponential scaling. This is quite unexpected and suggests the existence of a contraction-based polynomial-time simulation algorithm in this setup. An explicit description of such an algorithm could be interesting in its own right, and we leave this for future work.

\subsection{Random ZX-diagrams}

We simulate Erdős-Rényi random ZX-diagrams for different edge probabilities $p$ (Figure~\ref{bench-random-ZX-figure}). In particular, these ZX-diagrams are generated by creating $n$ green spiders with $\pi / 4$ phases and adding Hadamard edges between each pair of them with probability $p$. Since these ZX-diagrams do not originate from a circuit and thus have no causal flow/gflow, we can only run ``rw-greedy-linear'' and ``rw-greedy-b2t'' on these instances.

As the density of the ZX-diagram increases, Quimb's flops approach $2^n$. Notably, the performance of our rank-width--based methods is symmetric with respect to $p$. This is because cut-ranks are invariant under edge complement, which corresponds to flipping edge probabilities and switching between $G(n, p)$ and $G(n, 1 - p)$. Note that the maximum value of the blue curve is approximately $2^{n / 2}$, which aligns with the argument in Corollary~\ref{t-count-corollary} about the simulation complexity in terms of T-count.

\begin{figure}[htbp]
    \centering
    \begin{minipage}{0.45\textwidth}
        \centering
        \includegraphics[width=\textwidth]{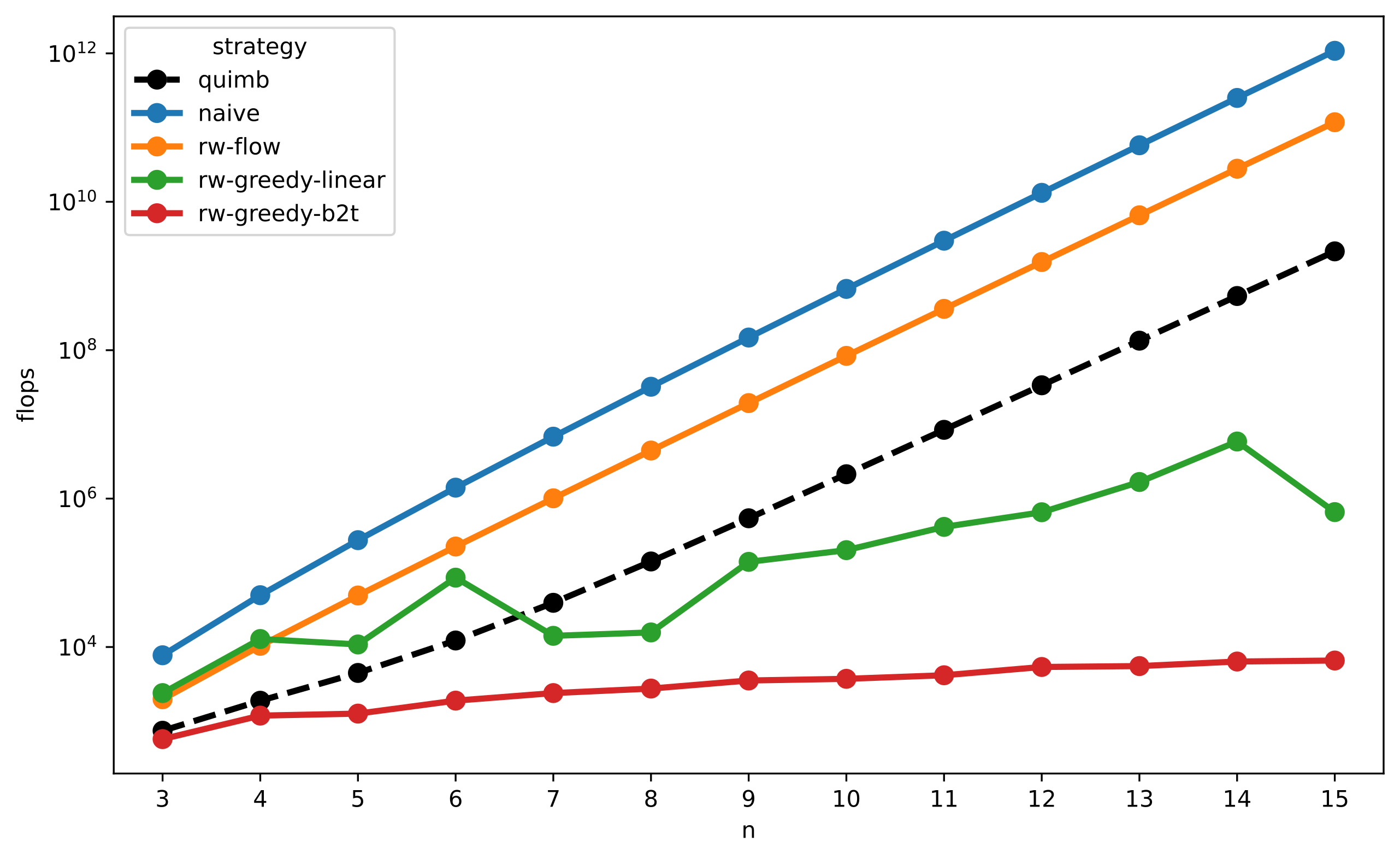}
        \caption{Simulating $n$-qubit Toffolis}
        \label{bench-tof-figure}
    \end{minipage}
    \ \ \ \
    \begin{minipage}{0.45\textwidth}
        \centering
        \includegraphics[width=\textwidth]{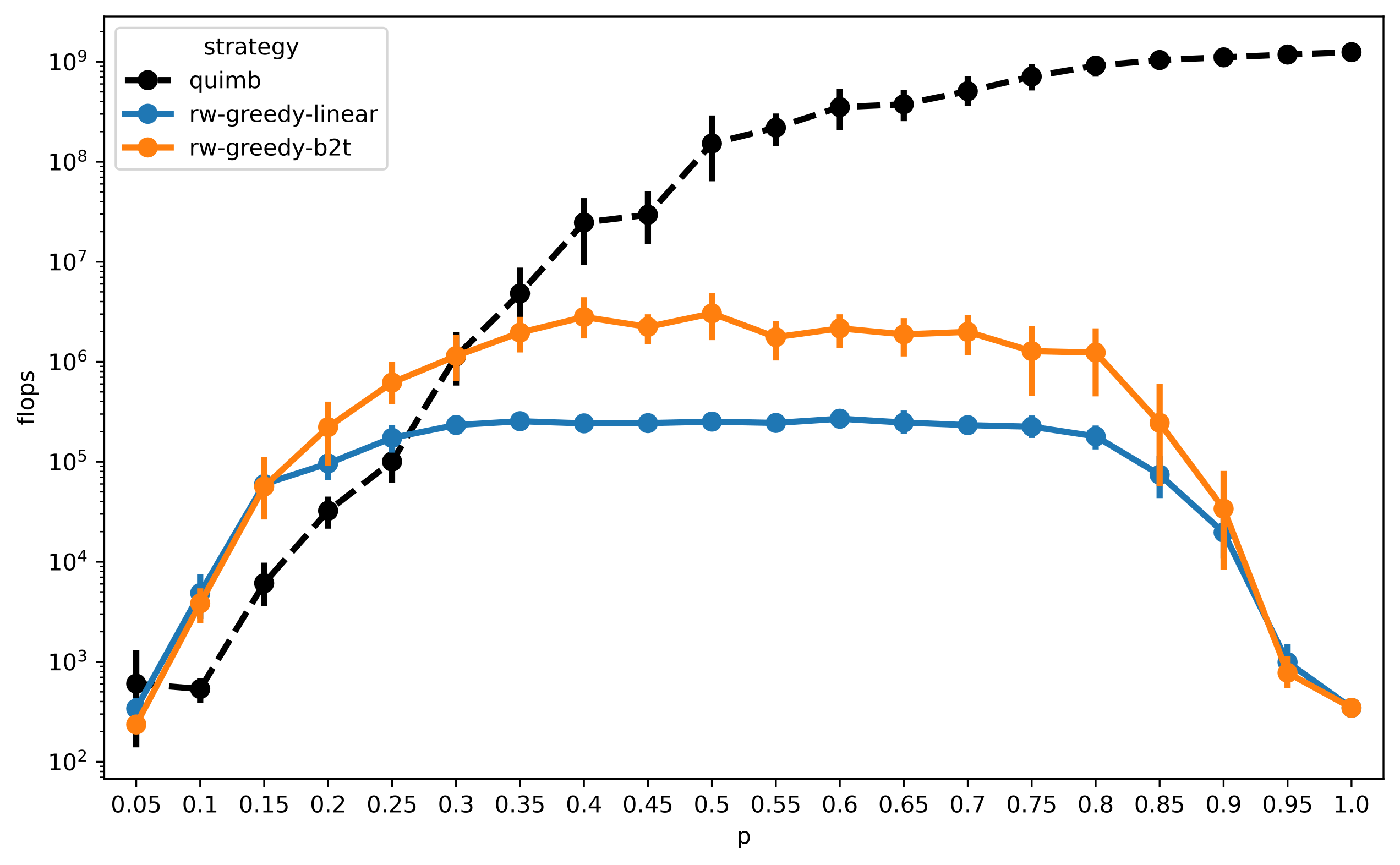}
        \caption{Simulating $G(30, p)$ ZX-diagrams}
        \label{bench-random-ZX-figure}
    \end{minipage}
\end{figure}

\section{Acknowledgements}

This work was supported by the Engineering and Physical Sciences Research Council grant number EP/Z002230/1, \textit{(De)constructing quantum software (DeQS)}. The authors would like to thank Julien Codsi and Tuomas Laakkonen for useful discussions on rank-decompositions and classical simulation.

\bibliographystyle{eptcs}
\bibliography{references.bib}

\appendix

\section{Rank-decomposition routines}
\label{rw-heur-appendix}

Here we provide the proofs, pseudocode, and complexity analysis of our rank-decomposition routines. Section~\ref{rw-flow-appendix} contains the proof of Lemma~\ref{gflow-rw-lemma}. In Sections~\ref{rw-greedy-linear-appendix},~\ref{rw-greedy-b2t-appendix}, we provide the pseudocode of ``rw-greedy-linear'' and ``rw-greedy-b2t'' and analyse their complexity.

\subsection{rw-flow}
\label{rw-flow-appendix}

\begin{proof}[Proof of Lemma~\ref{gflow-rw-lemma}]
    We prove that $p_i$ constitutes a linear rank-decomposition of width $\le |O|$ by induction from right to left. Let $M^{(i)}$ denote the adjacency matrix between $\{p_{i+1}, \dots, p_{|V|}\}$ and $\{p_1, \dots, p_i\}$, respectively, and set $r_i := \mathrm{rk}(M^{(i)})$. Let us show that for each $i \in [2; |V|]$, we have $r_{i-1} \le r_i$ for $p_i \not\in O$ and $r_{i-1} \le r_i + 1$ for $p_i \in O$.

    Note that $M^{(i-1)}$ can be obtained from $M^{(i)}$ by deleting its last column and appending a row, representing the edges between $p_i$ and $p_1, \dots, p_{i-1}$. It follows that $r_{i-1} \le r_i + 1$. We now assume $p_i \not\in O$ and show $r_{i-1} \le r_i$ by considering the following two cases. It then follows that $\max(r_i) \le |O|$, as desired.

    \begin{enumerate}[wide]
        \item $\lambda(p_i) = \mathrm{XY}$.
            By Definition~\ref{ext-gflow-definition}, we have $p_i \not\in g(p_i)$, $p_i \in \mathrm{Odd}(g(p_i))$, and for each $j < i$ we have $p_j \not\in \mathrm{Odd}(g(p_i))$. Take some $u \in g(p_i)$ and add all the rows from $g(p_i) \setminus \{u\}$ to the row $u$, so that it has a single one at the $i$-th column. Then, zero out the entire $i$-th column so that $M^{(i)}$ becomes block-diagonal. See Figure~\ref{adj-mat-XY-figure} for illustration. Now consider the submatrix $\widetilde M$ between $\{p_{i+1}, \dots, p_{|V|}\}$ and $\{p_1, \dots, p_{i-1}\}$, which corresponds to the `after--before' section in the first diagram of Figure~\ref{adj-mat-XY-figure}. Note that the row operations preserve the ranks of $\widetilde M$, $M^{(i)}$, and $M^{(i-1)}$. Due to the diagonalisation of $M^{(i)}$ as shown in the last diagram of Figure~\ref{adj-mat-XY-figure}, we have $\mathrm{rk}(\widetilde M) = \mathrm{rk}(M^{(i)}) - 1$. As $M^{(i-1)}$ differs from $\widetilde M$ by one extra row, we have $\mathrm{rk}(M^{(i-1)}) \le \mathrm{rk}(\widetilde M) + 1 = \mathrm{rk}(M^{(i)})$.
        \item $\lambda(p_i) \in \{\mathrm{XZ}, \mathrm{YZ}\}$. By Definition~\ref{ext-gflow-definition}, we have $p_i \in g(p_i)$, and for each $j < i$ we have $p_j \not\in \mathrm{Odd}(g(p_i))$. Hence, by adding the rows from $g(p_i) \setminus \{p_i\}$ to the row $p_i$, we can zero out its `before' part -- see Figure~\ref{adj-mat-Z-figure}. Define $\widetilde M$ as in the previous case. Clearly, $\mathrm{rk}(M^{(i-1)}) = \mathrm{rk}(\widetilde M) \le \mathrm{rk}(M^{(i)})$.
    \end{enumerate}

    \begin{figure}[htbp]
        \centering
        \begin{subfigure}[c]{0.25\textwidth}
            \includegraphics[width=\textwidth]{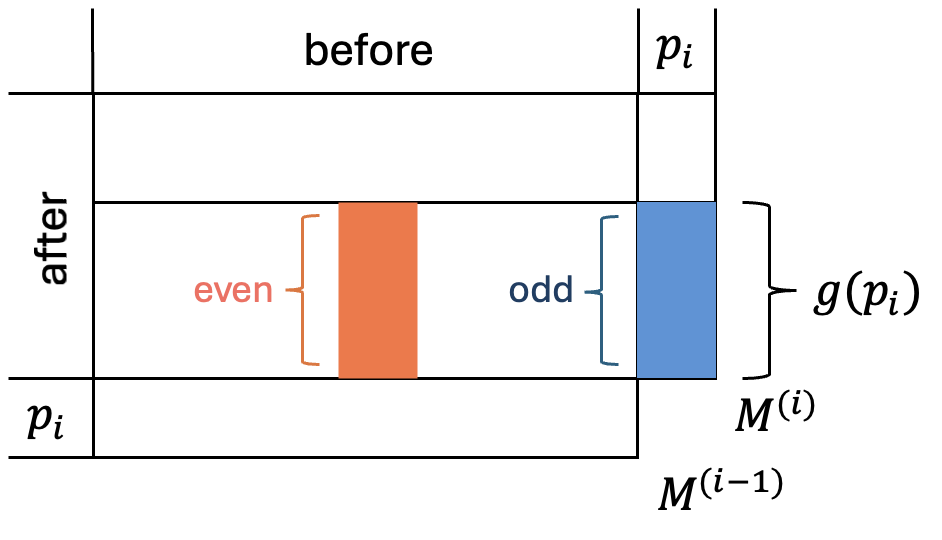}
        \end{subfigure}
        \begin{subfigure}[c]{0.05\textwidth}
            \includegraphics[width=\textwidth]{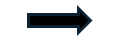}
        \end{subfigure}
        \begin{subfigure}[c]{0.25\textwidth}
            \includegraphics[width=\textwidth]{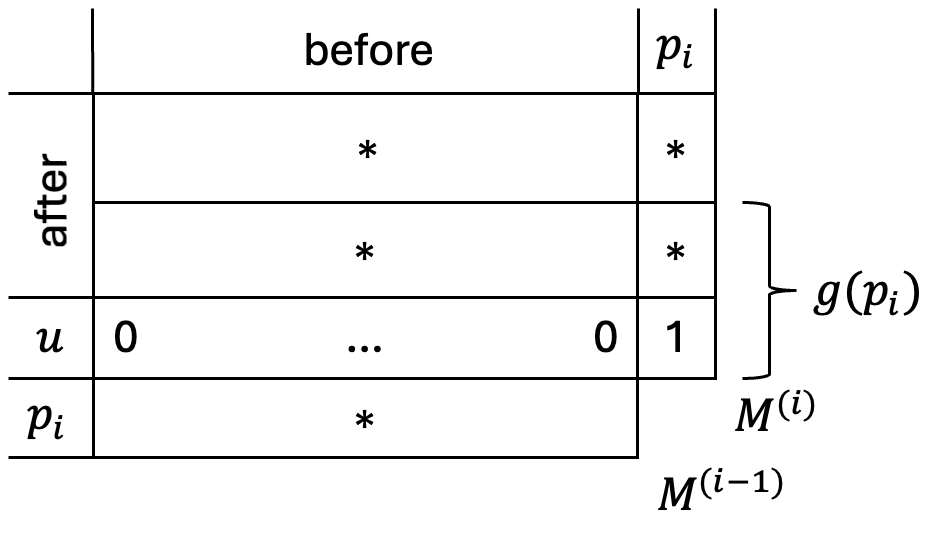}
        \end{subfigure}
        \begin{subfigure}[c]{0.05\textwidth}
            \includegraphics[width=\textwidth]{figures/RWFlow_right_arrow.png}
        \end{subfigure}
        \begin{subfigure}[c]{0.25\textwidth}
            \includegraphics[width=\textwidth]{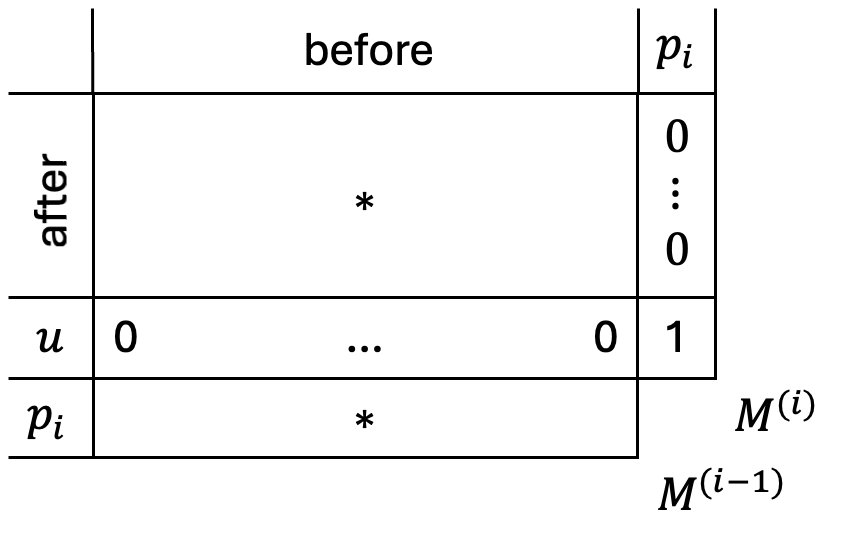}
        \end{subfigure}
        \caption{Block diagonalisation of the adjacency matrix in the XY case.}
        \label{adj-mat-XY-figure}
    \end{figure}

    \begin{figure}[htbp]
        \centering
        \begin{subfigure}[c]{0.3\textwidth}
            \includegraphics[width=\textwidth]{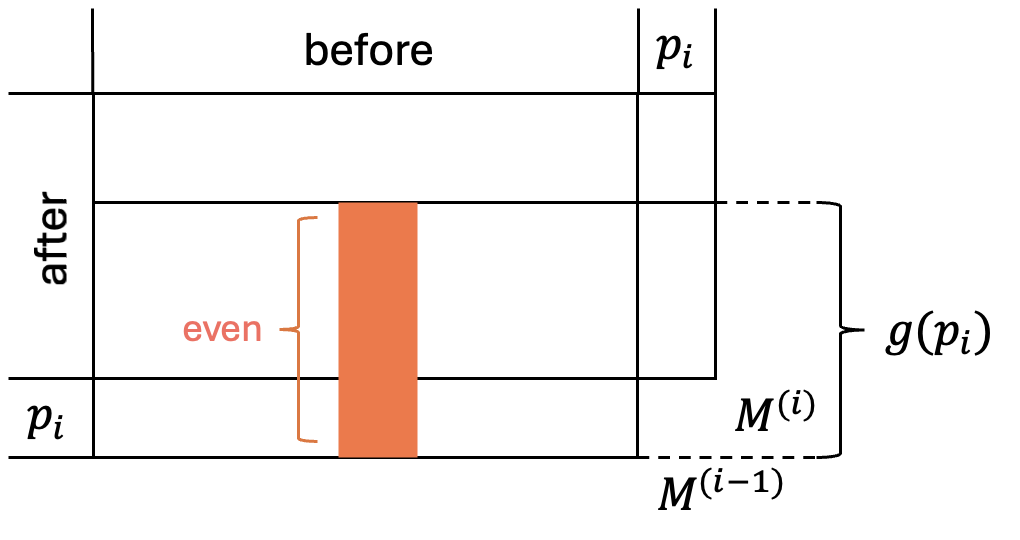}
        \end{subfigure}
        \begin{subfigure}[c]{0.05\textwidth}
            \includegraphics[width=\textwidth]{figures/RWFlow_right_arrow.png}
        \end{subfigure}
        \begin{subfigure}[c]{0.3\textwidth}
            \includegraphics[width=\textwidth]{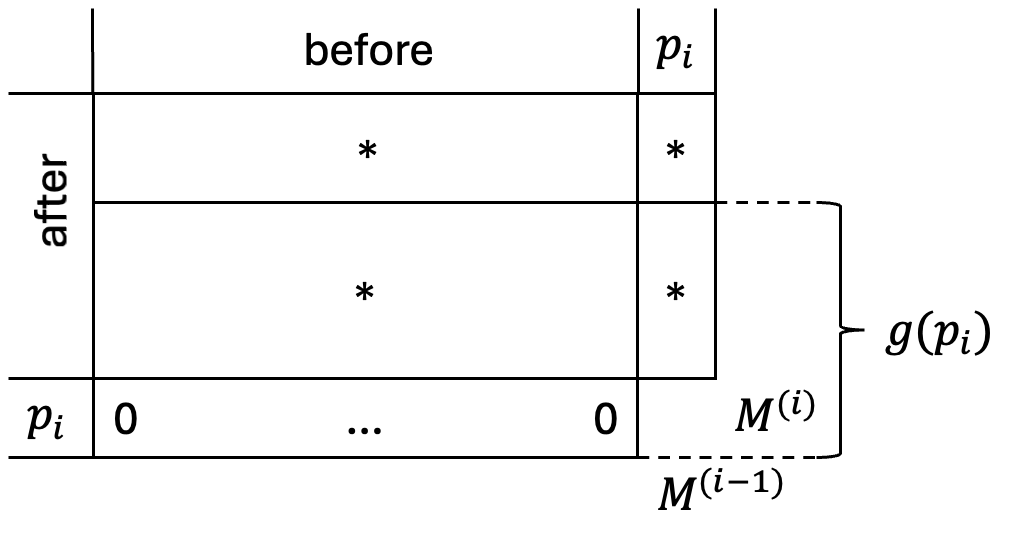}
        \end{subfigure}
        \caption{Matrix row reduction in the XZ and YZ cases.}
        \label{adj-mat-Z-figure}
    \end{figure}
\end{proof}

\subsection{rw-greedy-linear}
\label{rw-greedy-linear-appendix}

The pseudocode of ``rw-greedy-linear'' is given by Algorithm~\ref{naive-rw-greedy-linear-algorithm}. Note that each cut-rank calculation costs $O(|V|^3)$ operations (e.g. by doing Gaussian elimination). Since at each iteration there are $O(|V|)$ choices for $u$, the complexity of Algorithm~\ref{naive-rw-greedy-linear-algorithm} is $O(|V|^5)$. It is possible to reduce the complexity to $O(|V|^4)$, and we have implemented the optimised routine in PyZX.\footnote{\url{https://github.com/zxcalc/pyzx/blob/5f5e409/pyzx/rank\_width.py\#L172}}

\begin{algorithm}
    \caption{rw-greedy-linear}
    \begin{algorithmic}
        \Require $G = (V, E)$
        \Ensure $(p_i)_{i = 1, \dots, |V|}$ -- linear rank-decomposition
        \State $p \gets \mathrm{list}()$
        \For{$i = 1 \dots |V|$}
        \State $S \gets \{p_1, \dots, p_{i-1}\}$
        \State $M \gets \mathrm{Mat}_G(S, V \setminus S)$
        \If{$M = 0$}
        \State $p_i \gets \min(V \setminus S)$
        \Else
        \State $p_i \gets \mathrm{argmin}_{u \in \mathrm{pivotCols}(M)} \,\rho_G(S \cup \{u\})$
        \EndIf
        \EndFor
    \end{algorithmic}
    \label{naive-rw-greedy-linear-algorithm}
\end{algorithm}

\subsection{rw-greedy-b2t}
\label{rw-greedy-b2t-appendix}

The pseudocode of ``rw-greedy-b2t'' is given by Algorithm~\ref{naive-rw-greedy-b2t-algorithm}. Note that for each pair of subtrees $(S_u, S_v)$, computing the cut-rank $\rho_G(S_u \cup S_v)$ costs $O((|S_u| + |S_v|)|V|^2)$ operations. Summing over all $u, v$, each iteration cost is bounded by
\begin{align*}
    \sum_{u, v} O((|S_u| + |S_v|)|V|^2) = O\left(\sum_u |S_u| \cdot |V|^3\right) = O(|V|^4),
\end{align*}
hence the complexity of Algorithm~\ref{naive-rw-greedy-b2t-algorithm} is $O(|V|^5)$. Similarly to ``rw-greedy-linear'', it is possible to reduce the complexity to $O(|V|^4)$, and we have implemented the optimised routine in PyZX.\footnote{\url{https://github.com/zxcalc/pyzx/blob/5f5e409/pyzx/rank\_width.py\#L203}}

\begin{algorithm}
    \caption{rw-greedy-b2t}
    \begin{algorithmic}
        \Require $G = (V, E)$
        \Ensure $T$ -- rank-decomposition tree
        \State $S \gets \mathrm{dict}()$ \Comment{for each tree root in $T$, set of leaves in its component}
        \State $T \gets \mathrm{tree}()$ \Comment{forest of partial decompositions}
        \For{$i = 1 \dots |V|$}
        \State $S_i \gets \{i\}$
        \State $T.\mathrm{newVertex}()$ \Comment{initially, $T$ has $|V|$ isolated vertices numbered $1, \dots, |V|$}
        \EndFor
        \For{$t = 1 \dots |V| - 1$}
        \State $i, j \gets \mathrm{argmin}_{u, v:\,S_u \cap \mathrm{pivotCols}(\mathrm{Mat}_G(S_v, V \setminus S_v)) \neq \varnothing} \,\rho_G(S_u \cup S_v)$
        \State $k \gets T.\mathrm{newVertex}()$ \Comment{create a new root $k$ with children $i, j$}
        \State $T.\mathrm{addEdge}(k, i)$
        \State $T.\mathrm{addEdge}(k, j)$
        \State $S_k \gets S_i \cup S_j$ \Comment{update sets of leaves}
        \State $S.\mathrm{remove}(i)$
        \State $S.\mathrm{remove}(j)$
        \EndFor
    \end{algorithmic}
    \label{naive-rw-greedy-b2t-algorithm}
\end{algorithm}

\end{document}